\documentclass[twocolumn,10pt,letter]{asme2ej}

\usepackage{epsfig}
\usepackage{graphicx}
\usepackage{amsmath,amssymb, amsfonts}
\usepackage{multirow}
\usepackage{subfigure}

\usepackage{booktabs}
\usepackage{bbold}
\usepackage[ruled,vlined]{algorithm2e}
\usepackage{color}
\usepackage{mathtools}
\usepackage{newtxtext}
\usepackage{newtxmath}
\usepackage[hyphens]{url}
\usepackage[hidelinks]{hyperref}
\hypersetup{breaklinks=true}
\usepackage{setspace}
\usepackage{accents}
\newlength{\dhatheight}

\newcommand{\real}{\mathbb{R}}

\newtheorem{theorem}{Theorem}

\newtheorem{problem}{Problem}
\newtheorem{definition}{Definition}

\newtheorem{remark}{Remark}
\newcommand{\until}[1]{\{1,\dots, #1\}}

\newcommand{\supscr}[2]{#1^{\textup{#2}}}
\newcommand{\map}[3]{#1: #2 \rightarrow #3}

\newcommand{\param}[1]{\texttt{#1}}

\title{UAV Surveillance Under Visibility and \\Dwell-Time Constraints:\\A Sampling-Based Approach
\thanks{
This work has been sponsored by
the U.S.\ Army Research Office and the Regents of the University of
California, through Contract Number W911NF-09-D-0001 for the Institute for
Collaborative Biotechnologies, and that the content of the information does
not necessarily reflect the position or the policy of the Government or the
Regents of the University of California, and no official endorsement should
be inferred.}} 
\author{Jeffrey R. Peters\thanks{Corresponding author}\affiliation{United Technologies Research Center\\East Hartford, CT, CA 06118}}
\author{Amit Surana\affiliation{United Technologies Research Center\\East Hartford, CT 06118}}
\author{Grant S. Taylor\affiliation{Aviation Development Directorate\\United States Army\\Moffett Field, CA, 94035-1000}}
\author{Terry S. Turpin\affiliation{Aviation Development Directorate \\United States Army (Contractor)\\Moffett Field, CA, 94035-1000}}
\author{Francesco Bullo\affiliation{Department of Mechanical Engineering\\University of California\\Santa Barbara, CA 93106}}
\begin{document}
\maketitle
\vspace{-1em}
\begin{abstract}
  A framework is introduced for planning unmanned aerial vehicle flight paths for visual surveillance of ground targets, each having particular viewing requirements. Specifically, each target is associated with a set of imaging parameters, including a desired (i) tilt angle, (ii) azimuth, with the option of a 360-degree view, and (iii) dwell-time. Tours are sought to image the targets, while minimizing both the total mission time and the time required to reach the initial target. An $\epsilon$-constraint scalarization is used to pose the multi-objective problem as a constrained optimization, which, through careful discretization, can be approximated as a discrete graph-search. It is shown that, in many cases, this approximation is equivalent to a generalized traveling salesperson problem. A heuristic procedure for solving the discrete approximation and recovering solutions to the full routing problem is presented, and is shown to have resolution completeness properties. Algorithms are illustrated through numerical studies.\end{abstract}

\section{Introduction}
\label{sec:introduction}
The use of autonomous mobile sensors is becoming increasingly
common in civilian and military applications, since they can
provide support in tasks that are too dangerous, too expensive, or 
too difficult for humans to perform unaided. 
Tasks that can benefit from autonomous sensors include search and rescue, forest fire or oil spill monitoring,
surveillance and reconnaissance, transportation and
logistics management, and hazardous waste cleanup~\cite{JR:06,KPV:08, usaf:10}. These applications require intelligent and practical strategies to govern autonomous behavior in the presence of numerous constraints that arise in realistic missions.

This article considers a particular mobile sensor scenario that is of interest, e.g., in military operations~\cite{osd:05}, where a single fixed-wing Unmanned Aerial Vehicle (UAV) collects visual sensor data within a large geographic environment. Here, the UAV is equipped with a gimbaled camera, and must provide surveillance imagery of multiple static ground targets, each having associated imaging constraints. These pre-specified  constraints include (i) a desired tilt angle with tolerances, (ii) a desired azimuth with tolerances, including the option of a 360-degree view, and (iii) the amount of time that the UAV should dwell before moving to the next target. The goal is to construct flight paths that are optimal in some sense, while simultaneously 
allowing each target to be imaged with the prescribed specifications.

Ideally, we seek flight paths that simultaneously minimize two metrics: (i) the time required for the UAV to image all targets and return to the initial target, and (ii) the delay between the mission onset and the time the UAV reaches its first target. The latter goal is important because sensory data often plays a crucial role in the dynamic development of mission strategies, i.e., mission critical planning often cannot progress until some sensory
information is provided. This is true, for example, in reconnaissance missions where sensory data supports the operations of a manned aircraft, e.g.~\cite{USA:07}. 
Thus, it is often important that initial imagery is provided quickly, even if at the
expense of total tour duration.

Since these two objectives are conflicting in general, we construct solutions by constraining the second performance metric and optimizing over the first.
This is a standard \emph{scalarization} approach for addressing multi-objective problems, since solutions of the constrained problem allow for the construction of Pareto-optimal fronts~\cite{KM:98}. 
We develop a discrete-approximation strategy for constructing high-quality solutions to the scalarized problem, naturally leading to a complete heuristic framework for the full, multi-objective mission. This approach produces reasonable flight paths within a modular framework, while also leveraging existing solution strategies to classical routing problems to ensure straightforward, effective, practical implementation. 

Specifically, the contributions of this article are as follows. First,
we show how the multi-objective routing problem with both visibility
and dwell-time constraints can be posed as a constrained
optimization problem through reasonable assumptions about UAV
trajectories and dwell-time maneuvers. Indeed, we define
\emph{visibility regions} at each target that reflect imaging
requirements, along with a set of feasible dwell-time maneuvers, which are performed within the regions, that both accommodate the dynamic constraints of the UAV and are intuitively straightforward. These constructions are then translated
into a constraint set and incorporated into a precise problem
statement, which represents an $\epsilon$-constraint scalarization of
the multi-objective problem.  

Then, we illustrate how the constrained
optimization problem, having an infinite solution space, is approximated by a discrete problem, having a finite solution
space. This discrete approximation implicitly considers both the time
required for appropriate dwell-time maneuvers and the visibility
constraints, as it is formulated through a selective sampling
procedure. Namely, the UAV configuration space is carefully
sampled to obtain a discrete set whose elements are each paired with a
feasible dwell-time trajectory. Elements of this set form the nodes of a graph
whose edge weights are defined via an augmented Dubins distance that incorporates required dwell-times. The discrete approximation is defined as an optimal path-finding problem over the resultant
graph. 

Next, we present a novel heuristic method for solving the
discrete approximation that leverages solutions to standard
\emph{Generalized Traveling Salesperson Problem (GTSP)} instances. We
show that our heuristic method produces feasible solutions and, in
many cases, maps optimal solutions of a GTSP directly to
optimal solutions of the discrete problem. 

Finally, we incorporate
these constructions into a complete heuristic framework to produce
high-quality solutions to the full, multi-objective routing
problem. We prove that the heuristic is \emph{resolution complete} in
a specific mathematical sense, and finish by illustrating our methods
numerically.

We note that our work is most closely related to that of~\cite{KJO-PO-SD:12}, which uses a similar sampling-based framework in order to address a \emph{Polygon-Visiting Dubins Traveling Salesperson Problem (PVDTSP)}. However, the framework in~\cite{KJO-PO-SD:12} considers a single performance metric (total tour time), and cannot incorporate non-trivial UAV dwell-time behaviors. As such, our work is, in some sense, an expansion of~\cite{KJO-PO-SD:12} to incorporate a more general set of geometric and temporal imaging behaviors, and accommodate an additional performance metric that is reflective of realistic mission scenarios. In fact, the more general solution framework herein reduces to that of~\cite{KJO-PO-SD:12} in particular cases (Remark~\ref{rem:existing}). 

\section{Related Literature}
\label{sec:literature}
A large amount of
recent research focuses on the development of coordination strategies for UAV applications (e.g.~\cite{SR-RS-SD:05, DBK-RWB-RSH:08, JLF-PK-AG:15,VS-TS:15}). 
However, UAV research is a
sub-class of a much larger body of literature that addresses algorithmic
design and high-level reasoning for general mobile sensor applications~\cite{CAR-CYS-AT:07, JKH-BB-JL-BML:11}.
These applications often necessitate
solutions to complex routing problems, which may involve logical, temporal, and spatial
constraints, as well as environmental uncertainty.
When construction of global optima is not feasible, heuristics are often used to permit the real-time construction of
solutions for use within practical systems.
For example, variations on the classic
Traveling Salesperson Problem (TSP)~\cite{GG-APP:07, KS-EF-FB:06h,
  JJE-KS-EF-FB:08k} arise frequently and, since the TSP is NP-hard, global optima usually cannot be consistently found in reasonable time. However, several new
insights have been developed over the last decade for the classical TSP, along with a number of variations~\cite{WM-SR-SD:07,
  MN-PTK-ARG:15, MA-JPH:08}, and sophisticated heuristic solvers, e.g.~\cite{KH:00} can quickly construct high-quality solutions for practical use.
 Other problems that are common in mobile sensor applications include the construction region selection policies~\cite{VS-FP-FB:11za, RP-PA-FB:14b}, the derivation of static and
dynamic coverage schemes~\cite{ym-ak:08, FB-EF-MP-KS-SLS:10k}, the development of persistent task execution frameworks~\cite{NM-SLS-SLW:15}, and the design of load balancing
strategies~\cite{FB-RC-PF:08u, JGC:12}.

The problem herein is loosely interpreted as a
generalization of both the \emph{Polygon-Visiting Dubins Traveling Salesperson
  Problem (PVDTSP)}~\cite{KJO-PO-SD:12, KJO:09} and the \emph{Dubins Traveling Salesperson Problem with Neighborhoods (DTSPN)}~\cite{JTI-JPH:13}. The PVDTSP and the DTSPN are variations on the classic \emph{Dubins TSP (DTSP)} that require vehicles to visit a series of geometric regions rather than discrete points.  
To incorporate explicit imaging constraints within a multi-objective framework, we adopt a strategy that
is, in some sense, an extension of~\cite{KJO-PO-SD:12}, where the
authors approximate solutions to a PVDTSP by discretizing
regions of interest and posing the resulting problem as a \emph{Generalized
  Traveling Salesperson Problem (GTSP)} (also called the \emph{Set TSP}, \emph{Group TSP}, 
\emph{(Finite) One-in-a-set TSP}, \emph{Multiple Choice TSP}, or \emph{Covering Salesperson Problem}). Sampling-based strategies that approximate a continuous motion planning problem with a discrete path-finding problem over a graph are typically called \emph{sampling-based roadmap
  methods}~\cite[Ch. 5]{SML:06}. Such methods are
traditionally used for point-to-point planning among obstacles; however, they have also been used for more general Dubins path planning, e.g.~\cite{PO-SR-SD:09b}.

Like the TSP, the GTSP is a combinatorial optimization; however, strategies exist for computing high-quality
solutions. The most popular solution approach converts the GTSP into an \emph{Asymmetric Traveling
Salesperson Problem} (ATSP) via a \emph{Noon-Bean transform}~\cite{CEN-JCB:91a}. The availability of efficient solvers for the ATSP, e.g. LKH~\cite{KH:00}, make such transformations a viable GTSP solution option in practice, though they do not produce global optima in general. 
Other common GTSP solution approaches make use of meta-heuristics, e.g.~\cite{LVS-MSD:00}.
%
%

Aside from the explicit visibility and dwell-time constraints, the problem herein differs from typical ``TSP-like'' problems due to the presence of multiple (potentially conflicting) performance objectives.
A vast amount of literature addresses multi-objective optimization problems, some of which focuses
on purely mathematical formulations while some targets solutions for
particular applications~\cite{RTM-AJS:04,
  PJF-RCP-RJL:05,DAG-DEJ:04}. Solution approaches typically
seek to satisfy a pre-determined notion of optimality, the
most common being \emph{Pareto optimality}. Pareto-optimal fronts are usually difficult to characterize directly, so they are typically constructed using  \emph{scalarization} methods, which reduce the multi-objective problem to a single-objective problem whose optimal solutions map to Pareto-optimal solutions with respect to the original problem. The most common scalarization techniques are \emph{linear scalarization}, where the cost is a linear combination of the objectives, and \emph{$\epsilon$-constraint method}, where the values of all but a single objective are explicitly treated as optimization constraints~\cite{KM:98}. 

Our approach to the present problem uses an $\epsilon$-constraint method that treats the initial UAV maneuver time as an explicit optimization constraint.  
In general, no single approach to multi-objective
optimization is always superior; the
appropriate method depends on the type of information available, the user's preferences, solution requirements, and the
availability of software~\cite{PJF-RCP-RJL:05}. We choose the $\epsilon$-constraint method since it (i) is a standard approach to multi-objective optimization, (ii) gives explicit flexibility in the degree to which the initial maneuver time is considered, allowing the approach to be easily tailored to specific mission scenarios, (iii) allows for a heuristic solution procedure that is naturally modular and can leverage solutions to standard routing problems, (iv) promotes straightforward approximation of Pareto-optimal fronts by varying scalarization parameters, and (v) is straightforward and intuitive, making it an attractive approach to both practitioners and theoreticians alike.

\section{Problem Formulation}
\label{sec:formulation}
\subsection{UAV Specifications}
\label{sec:UAV_specs}
Consider a single fixed-wing UAV, equipped with a GPS location device and a gimbaled, omnidirectional camera. 
The camera is steered by
a low-level controller, which is independent of the vehicle motion
controller. For simplicity, we neglect the possibility of camera occlusions.  
The work herein focuses on high-level UAV trajectory planning, rather than low-level vehicle or camera motion control. 
We consider planar motion in a global, ground-plane reference frame, assuming the UAV maintains a fixed altitude $a$ and a fixed speed $s$. 
We model the UAV as a 
Dubins vehicle~\cite{LED:57} with minimum turning radius $r$, 
neglecting dynamic effects caused by
wind, etc.
Let $v_0 \in \real^2 \times [0,2\pi)$ denote the UAV's initial configuration (location, heading).

\subsection{Target Specifications}
\label{sec:target}
Consider $M$ static targets, each with associated imaging, i.e., visibility and dwell-time, requirements.
Each target $T_j$ is associated with the following (fixed) parameters:
\begin{enumerate}
\item $t_j \in \real^2$, the location of the target in the ground-plane reference frame,
\item $\param{BEH}_j \in \{\param{ANY}, \param{ANGLE}, \param{FULL}\}$, the required viewing behavior, where $\param{ANY}$ indicates no preference for the azimuth of the collected images, $\param{ANGLE}$ indicates that the target should be imaged at a specific azimuth, and $\param{FULL}$ indicates that a 360-degree view of the target should be provided,
\item $\left[\supscr{\phi_j}{A}-\supscr{\Delta_j}{A},\supscr{\phi_j}{A}+\supscr{\Delta_j}{A}\right] \subset \real$, a range of acceptable azimuths when $\param{BEH}_j = \param{ANGLE}$, as measured with respect to a reference ray in the ground plane (Fig.~\ref{fig:parameters}, top), 
\item $\left[\supscr{\phi_j}{T} - \supscr{\Delta_j}{T}, \supscr{\phi_j}{T}+\supscr{\Delta_j}{T}\right] \subset (0,\frac{\pi}{2}]$, acceptable camera tilt angles as measured with respect to a plane parallel to the ground-plane (Fig.~\ref{fig:parameters}, bottom), and
\item  $\tau_j \in \mathbb{Z}_{\geq 0}$, the required number of dwell-time ``loops.''
\end{enumerate}
\begin{figure}
\centering
\includegraphics[width = 0.56\columnwidth]{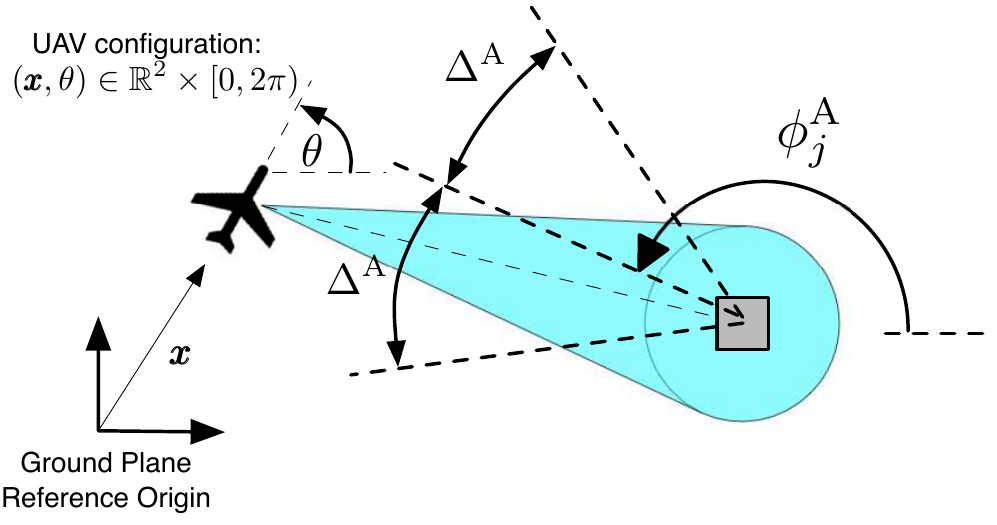}\\

\vspace{5mm}
\includegraphics[width = 0.65\columnwidth]{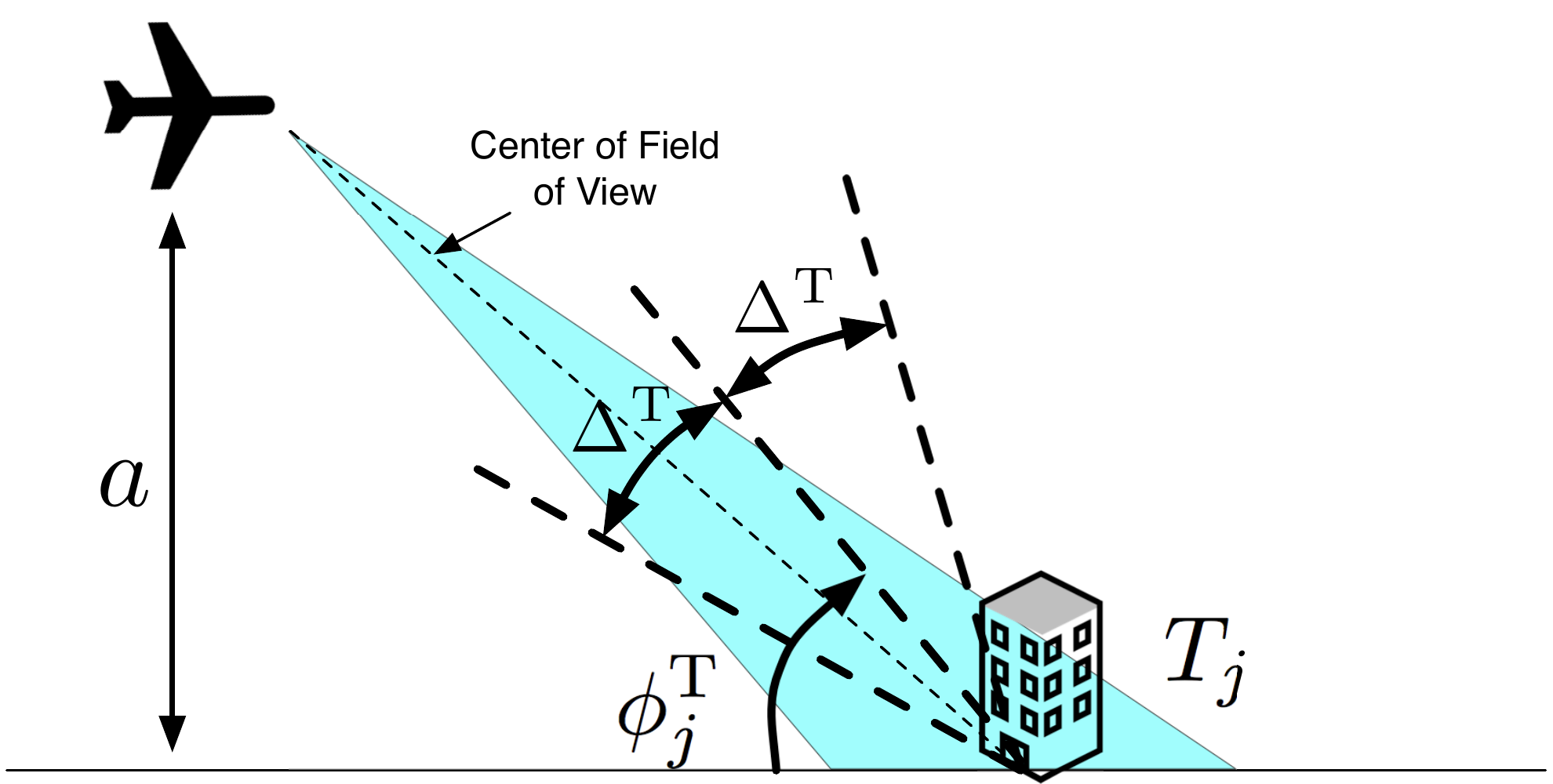} 
\caption{Illustration of key imaging parameters associated target $T_j$.  }
\label{fig:parameters}
\end{figure}
Define the \emph{visibility region}, $\texttt{VIS}_j \subset \real^2$, for target $T_j$ as the set of locations from which the UAV is able to image the target with an acceptable tilt angle and azimuth (that is, a tilt angle within the interval $\left[\supscr{\phi_j}{T} - \supscr{\Delta_j}{T}, \supscr{\phi_j}{T}+\supscr{\Delta_j}{T}\right]$ and, if $\param{BEH}_j = \param{ANGLE}$, an azimuth within the interval $\left[\supscr{\phi_j}{A}-\supscr{\Delta_j}{A},\supscr{\phi_j}{A}+\supscr{\Delta_j}{A}\right]$; if $\param{BEH}_j \neq \param{ANGLE}$, then any azimuth is acceptable).
Each $\texttt{VIS}_j$ is uniquely defined by the UAV altitude $a$, the location $t_j$, the behavior $\param{BEH}_j$, and the intervals $\left[\supscr{\phi_j}{T} - \supscr{\Delta_j}{T}, \supscr{\phi_j}{T}+\supscr{\Delta_j}{T}\right]$, $\left[\supscr{\phi_j}{A}-\supscr{\Delta_j}{A},\supscr{\phi_j}{A}+\supscr{\Delta_j}{A}\right]$. Algorithm~\ref{alg:visibility} presents the methodology for constructing visibility regions.
Note that, if $\param{BEH}_j \neq \param{ANGLE}$, then $\texttt{VIS}_j$ is a full annulus centered at $t_j$; otherwise, $\texttt{VIS}_j$ is an annular sector (Fig.~\ref{fig:visibility}). 
\begin{figure}
\centering
\includegraphics[width = 0.38\columnwidth]{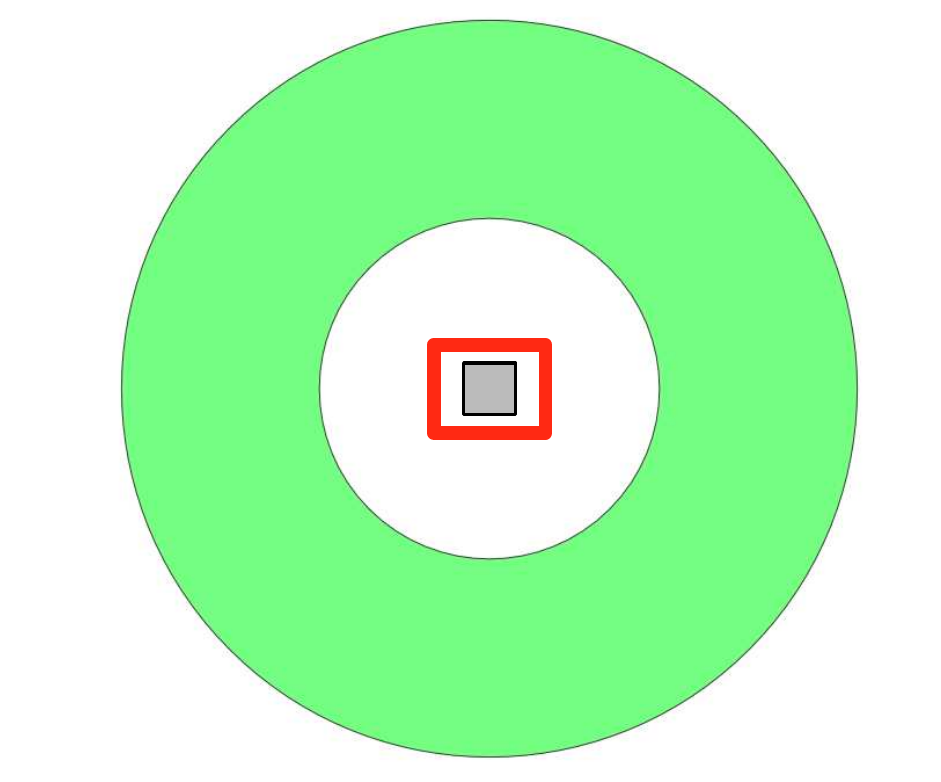}\hspace{3mm}\vrule\hspace{3mm}
\includegraphics[width = 0.43\columnwidth]{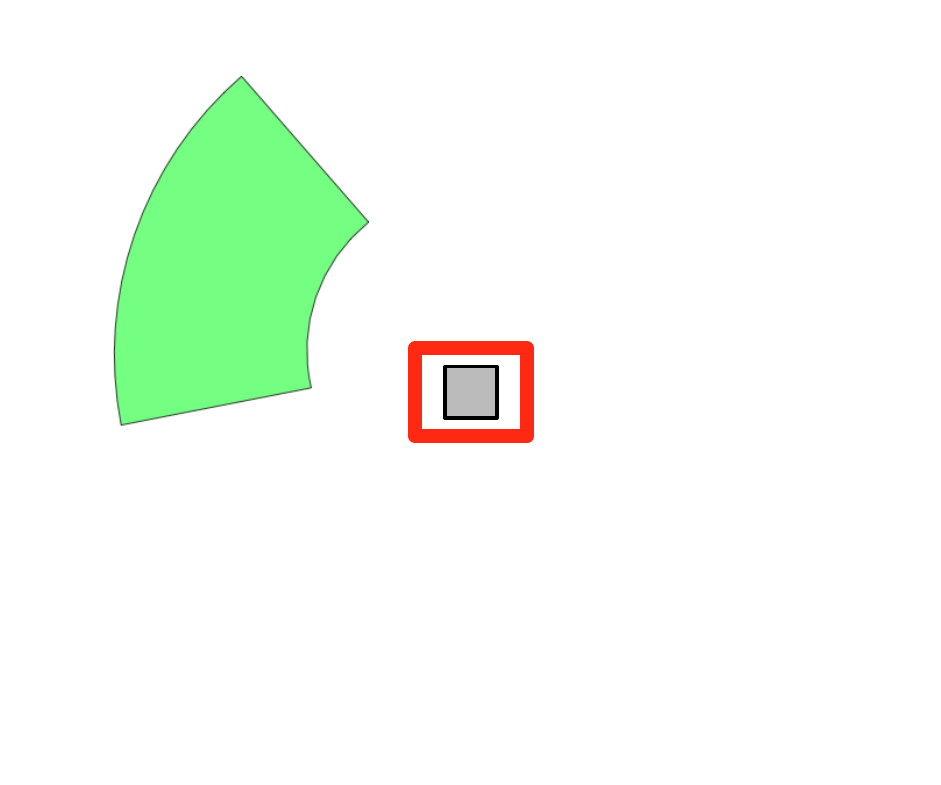} 
\caption{An example visibility region $\texttt{VIS}_j$ when $\param{BEH}_j \neq \param{ANGLE}$ (left), and when $\param{BEH}_j = \param{ANGLE}$ (right).}
\label{fig:visibility}
\end{figure}
As such, fixing target locations, each visibility region is parameterized by two radii (lower, upper radial limits) together with two angles (lower, upper angular limits).
%
\SetAlgorithmName{Algorithm}{ of Algorithms}{Algorithms}
\IncMargin{.0em}
\begin{algorithm}[h]
 {\footnotesize
  \SetKwInOut{Input}{Input}
   \SetKwInOut{Output}{Output}
  \SetKwInOut{Set}{Set}
  \SetKwInOut{Title}{Algorithm}
  \SetKwInOut{Require}{Require}
  \Input{a; $\supscr{\phi_j}{T}, \supscr{\phi_j}{A}, \supscr{\Delta_j}{T}, \supscr{\Delta_j}{A}$ for each $j \in \until M$} 
  \Output{$\{\texttt{VIS}_j\}_{j \in \until M}$}
  \BlankLine
  \For{ Each $T_j$}{
  	\eIf{$\param{BEH}_j \neq \param{ANGLE}$}{
		\nl Define $\texttt{VIS}_j$ as the annulus in $\real^2$ centered at $t_j$ with\\ \hspace{3mm} radial limits $a/\tan(\supscr{\phi_j}{T}\pm\supscr{\Delta_j}{T})$
	}{\nl Define $\texttt{VIS}_j$ as the annular sector in $\real^2$ centered at $t_j$  \\ \hspace{3mm} with radial limits $a/\tan(\supscr{\phi_j}{T}\pm\supscr{\Delta_j}{T})$, and angular  \\ \hspace{3mm} limits equivalent to $\supscr{\phi_j}{A}\pm\supscr{\Delta_j}{A}$.
	}
  }
\nl \Return $\{\texttt{VIS}_j\}_{j \in \until M}$
    \caption{\textit{Visibility Region Construction}}
 \label{alg:visibility}}
\end{algorithm} 
\DecMargin{.0em}

The variable $\tau_j \in \mathbb{Z}_{\geq 0}$ indicates the number of dwell-time ``loops'' that are required at the target $T_j$. 
If $\tau_j = 0$, then the UAV accomplishes its task by passing over any point within $\texttt{VIS}_j$. If $\tau_j >0$, i.e., non-trivial dwell time is specified, 
assume the UAV images $T_j$ as follows: If  $\param{BEH}_j = \param{FULL}$, the UAV makes $\tau_j$ full circles around the target location at some constant radius;
if $\param{BEH}_j \neq \param{FULL}$, then the UAV selects a pivot point within $\texttt{VIS}_j$ and makes $\tau_j$ circles about the selected point at radius $r$. 
Each non-trivial dwell-time maneuver must be performed entirely within the appropriate visibility region.
Fig.~\ref{fig:path} shows examples of acceptable imaging behavior for various choices of $\tau_j$ and $\param{BEH}_j$. 
\begin{figure*}
\centering
\includegraphics[width = 0.19\textwidth]{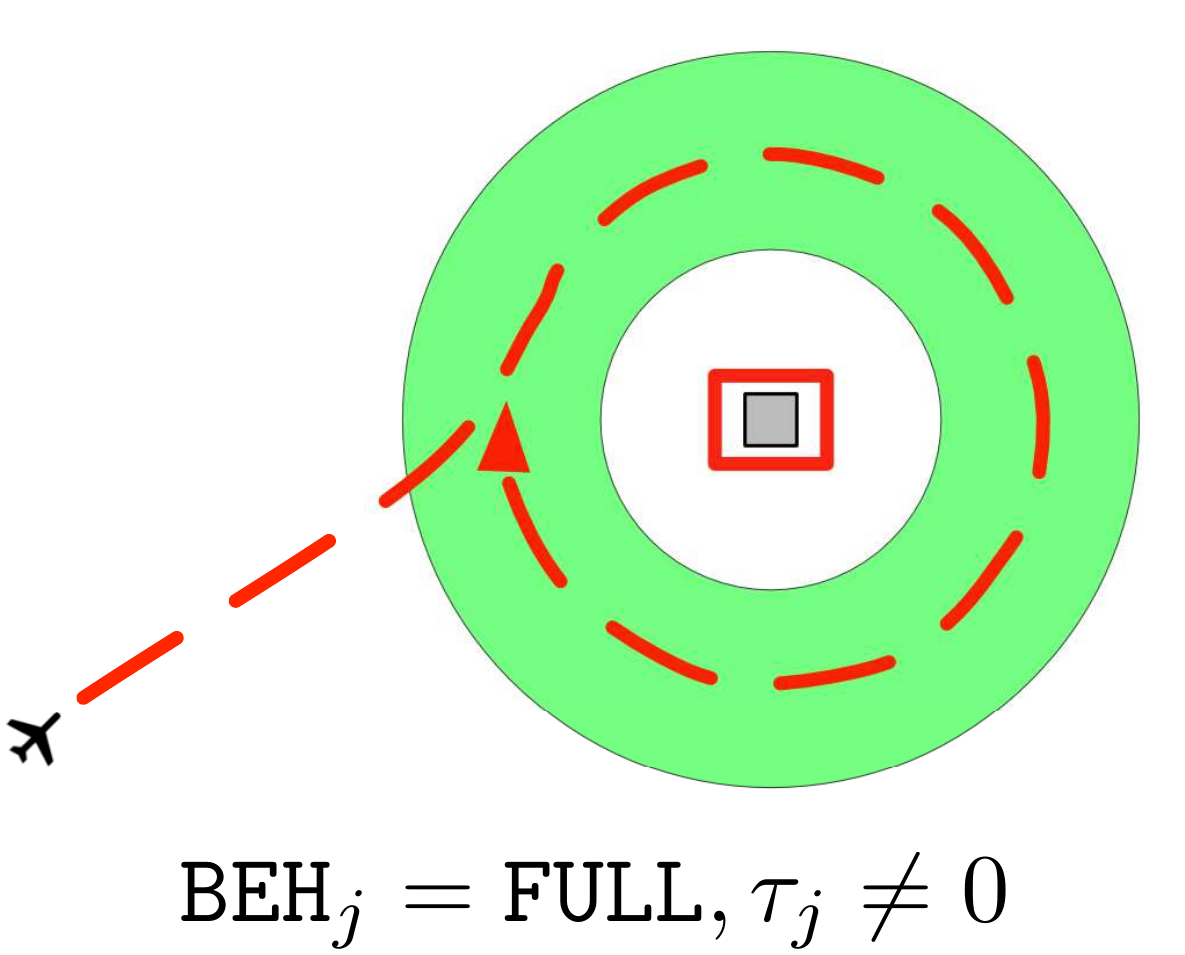} \vrule
\includegraphics[width = 0.19\textwidth]{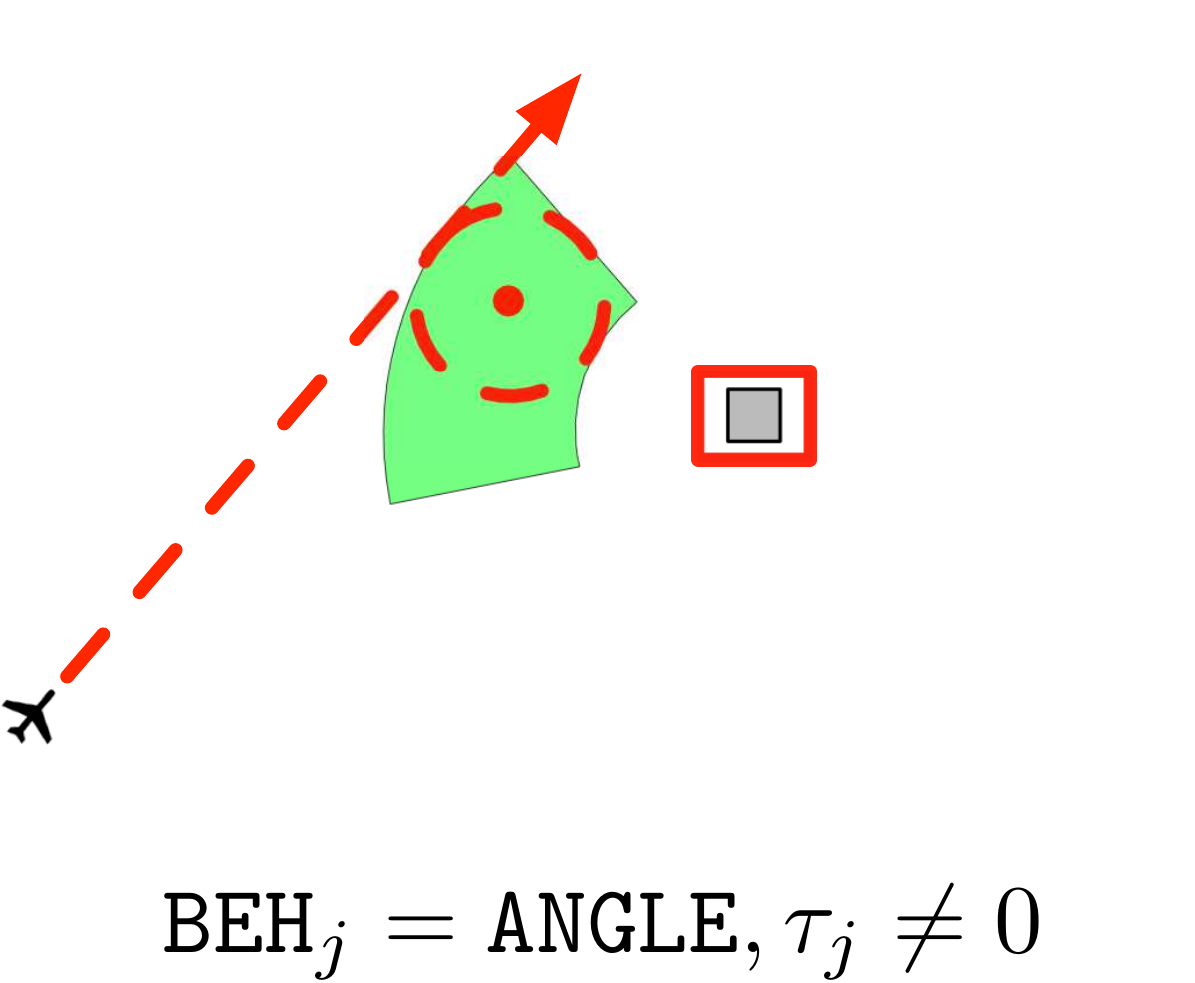}\vrule
\includegraphics[width = 0.19\textwidth]{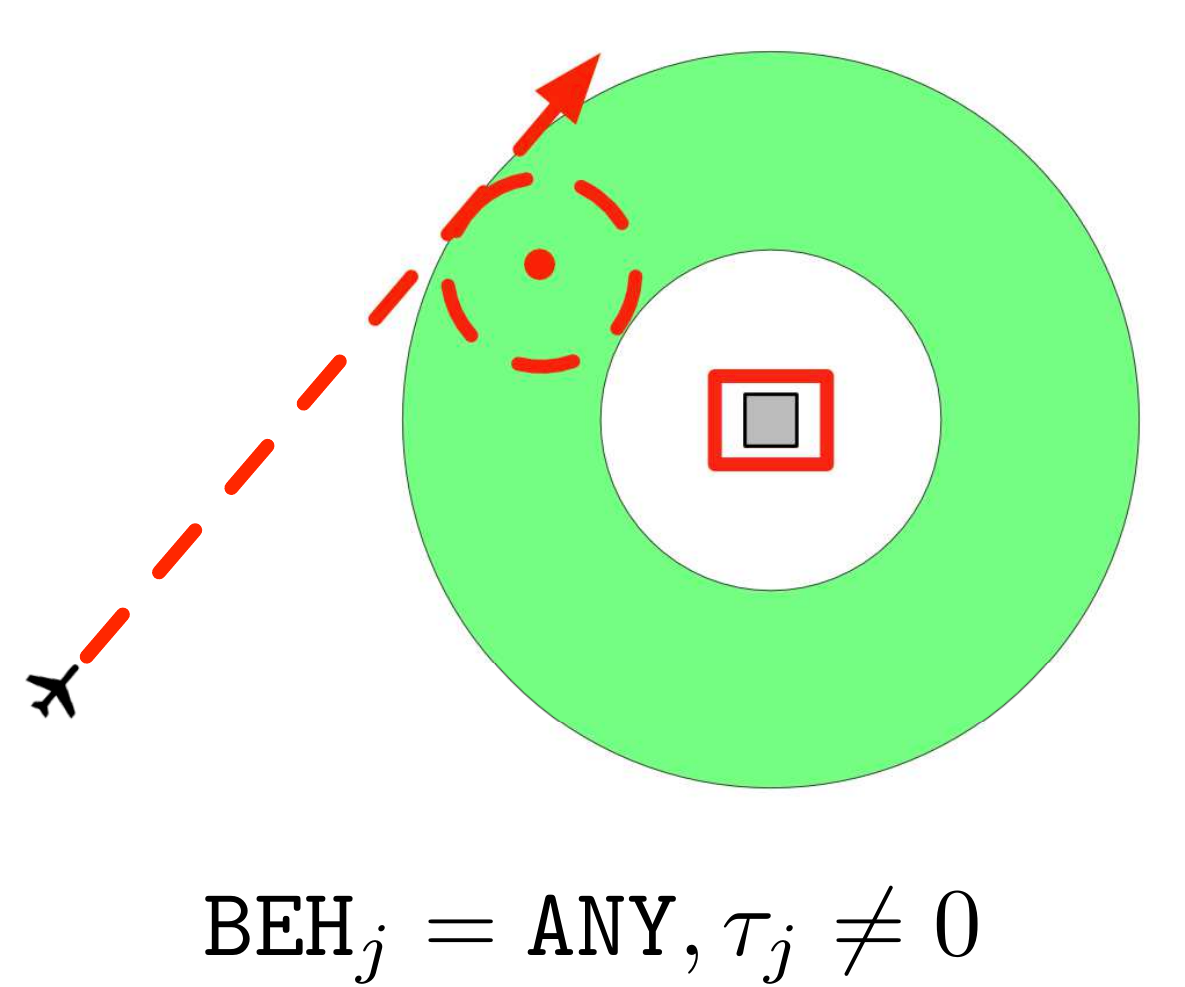}\vrule
\includegraphics[width = 0.19\textwidth]{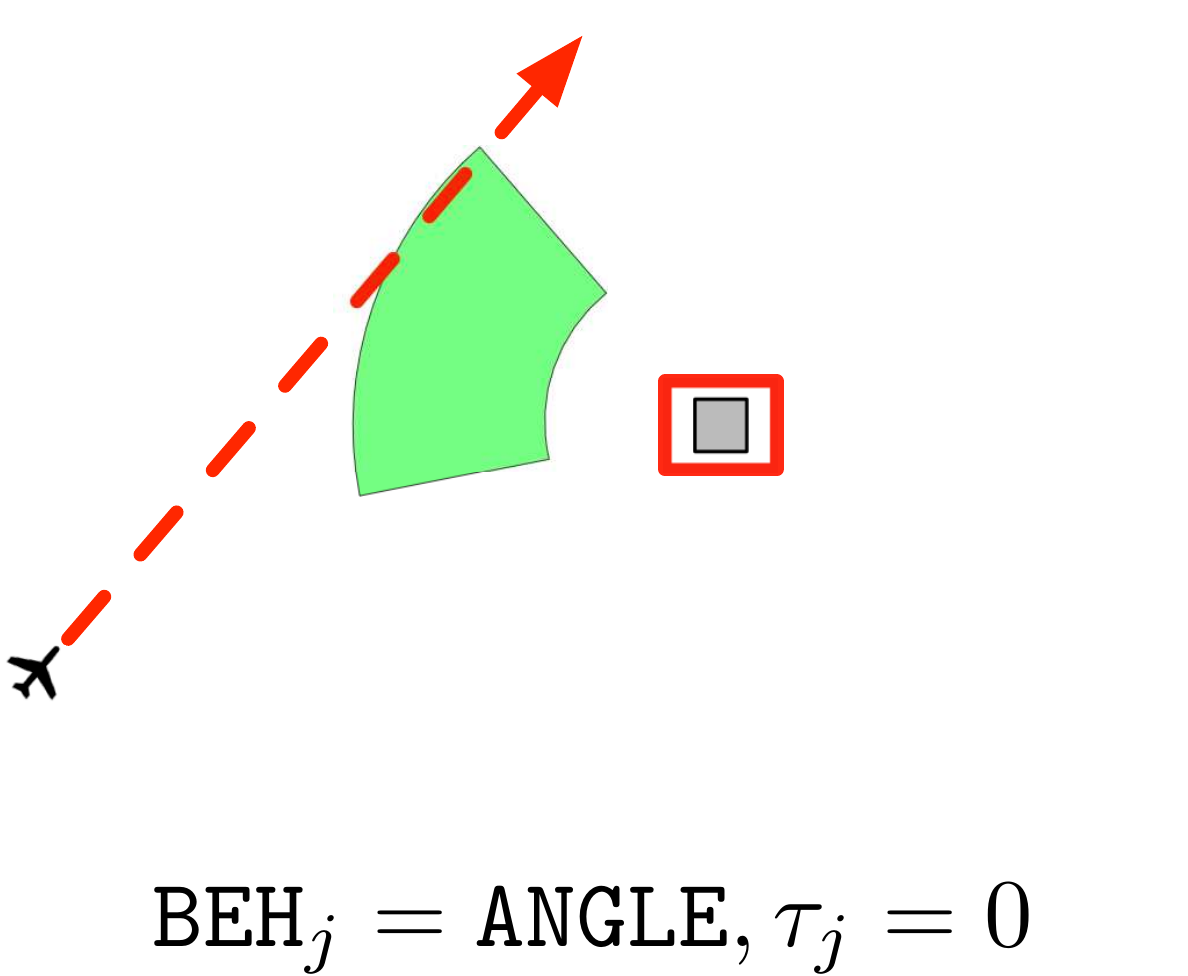}\vrule
\includegraphics[width = 0.19\textwidth]{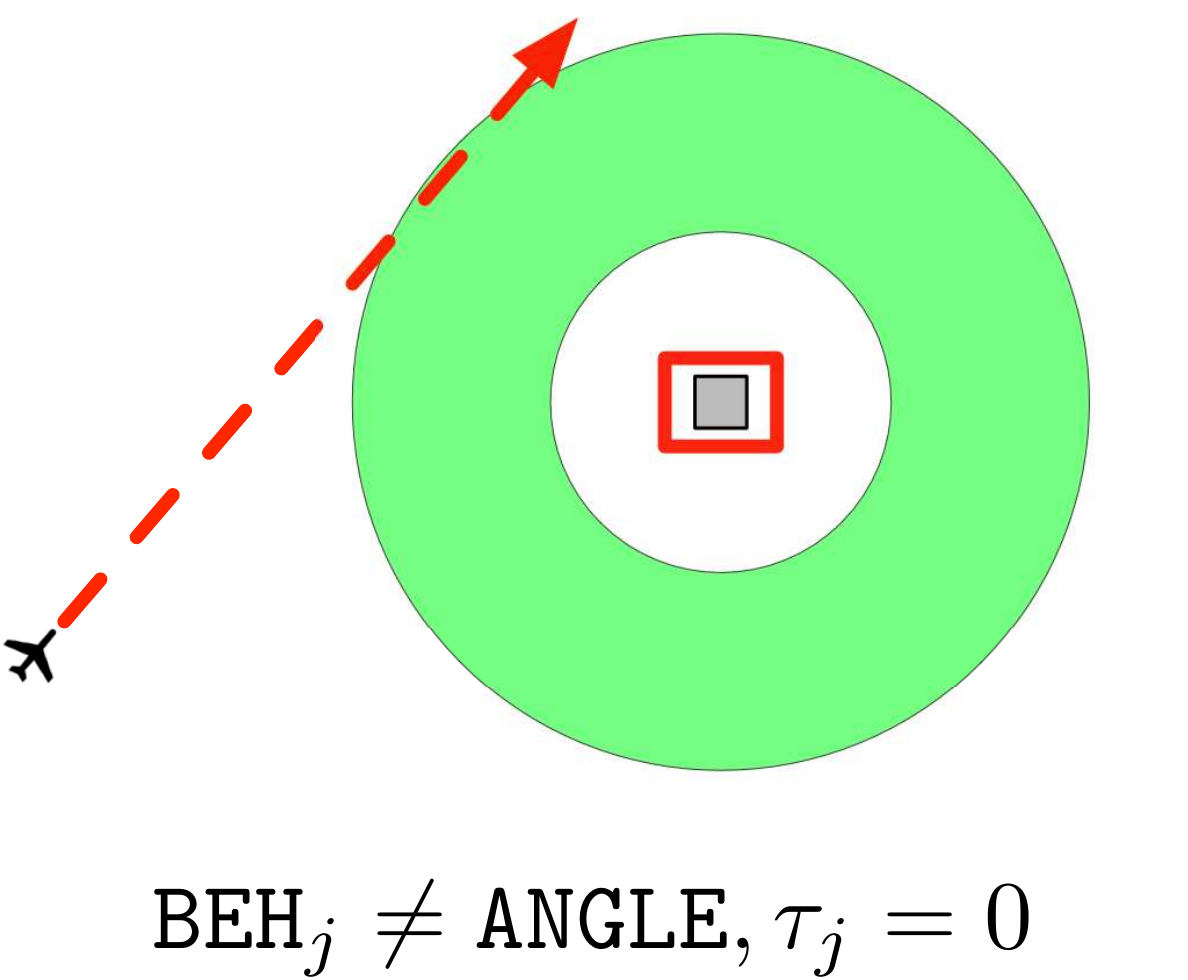}
\caption{Example imaging behaviors at target $T_j$ for various choices of $\param{BEH}_j$ and $\tau_j$.}
\label{fig:path}
\end{figure*}
\begin{figure*}
\centering
\includegraphics[width = 0.19\textwidth]{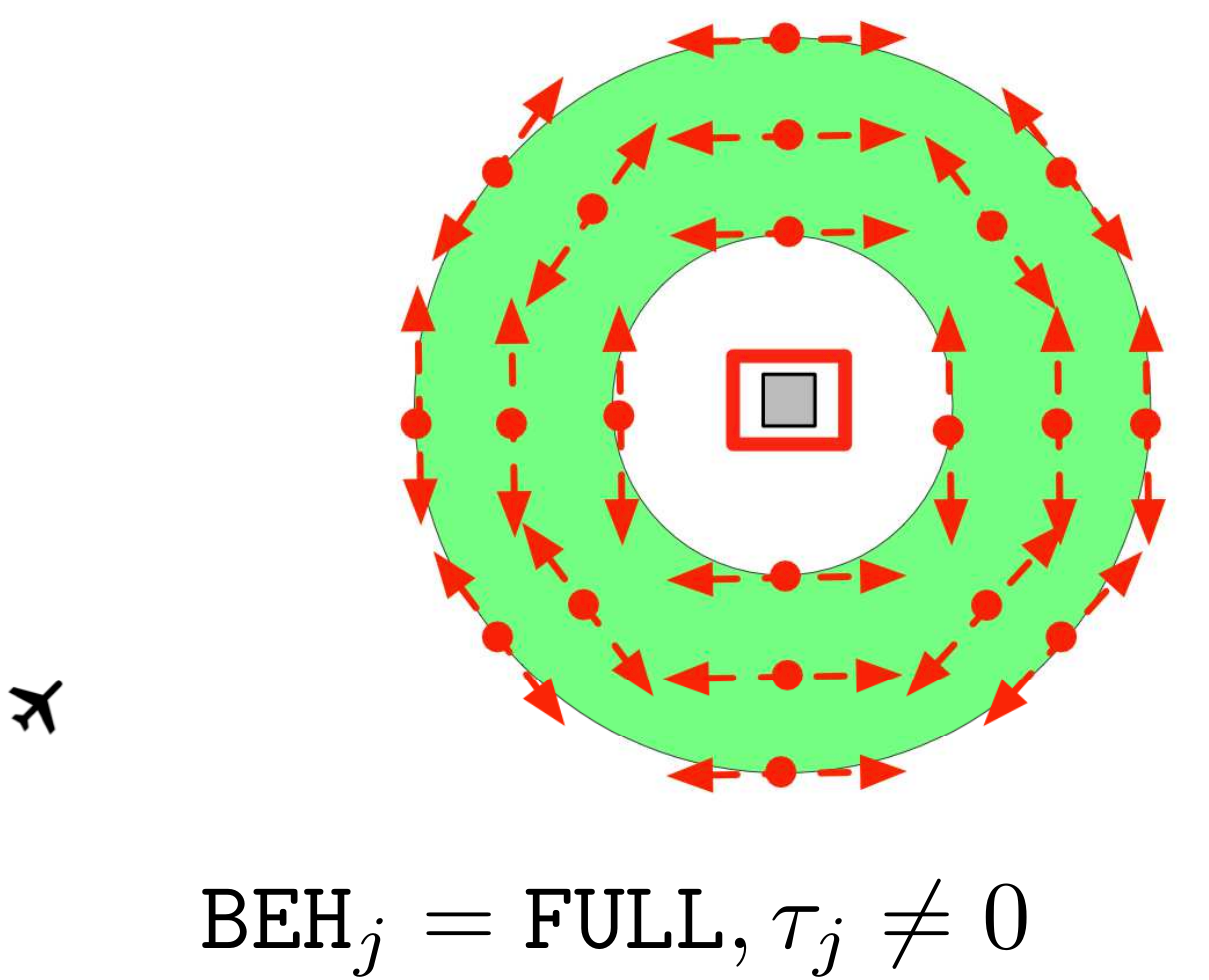}\vrule
\includegraphics[width = 0.19\textwidth]{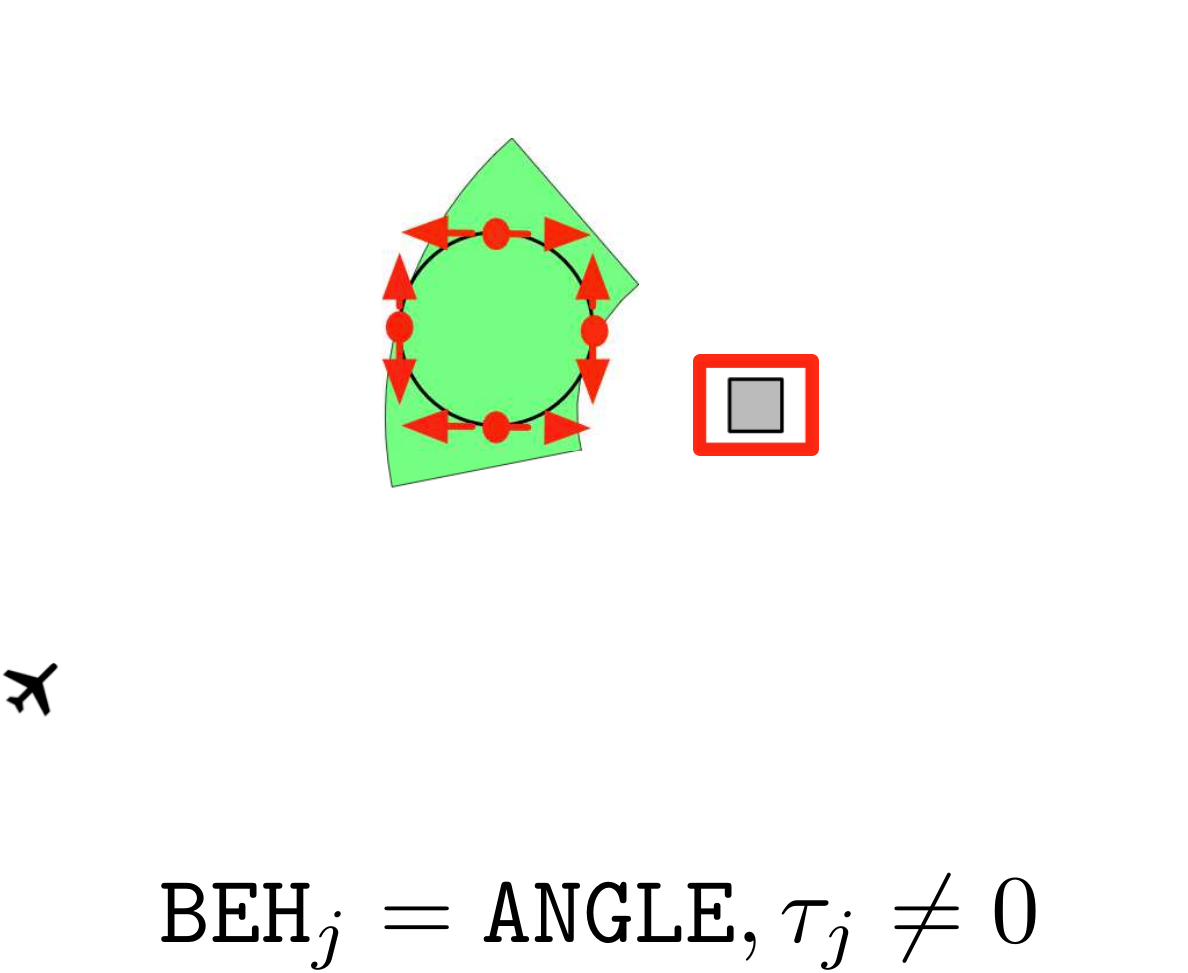}\vrule
\includegraphics[width = 0.19\textwidth]{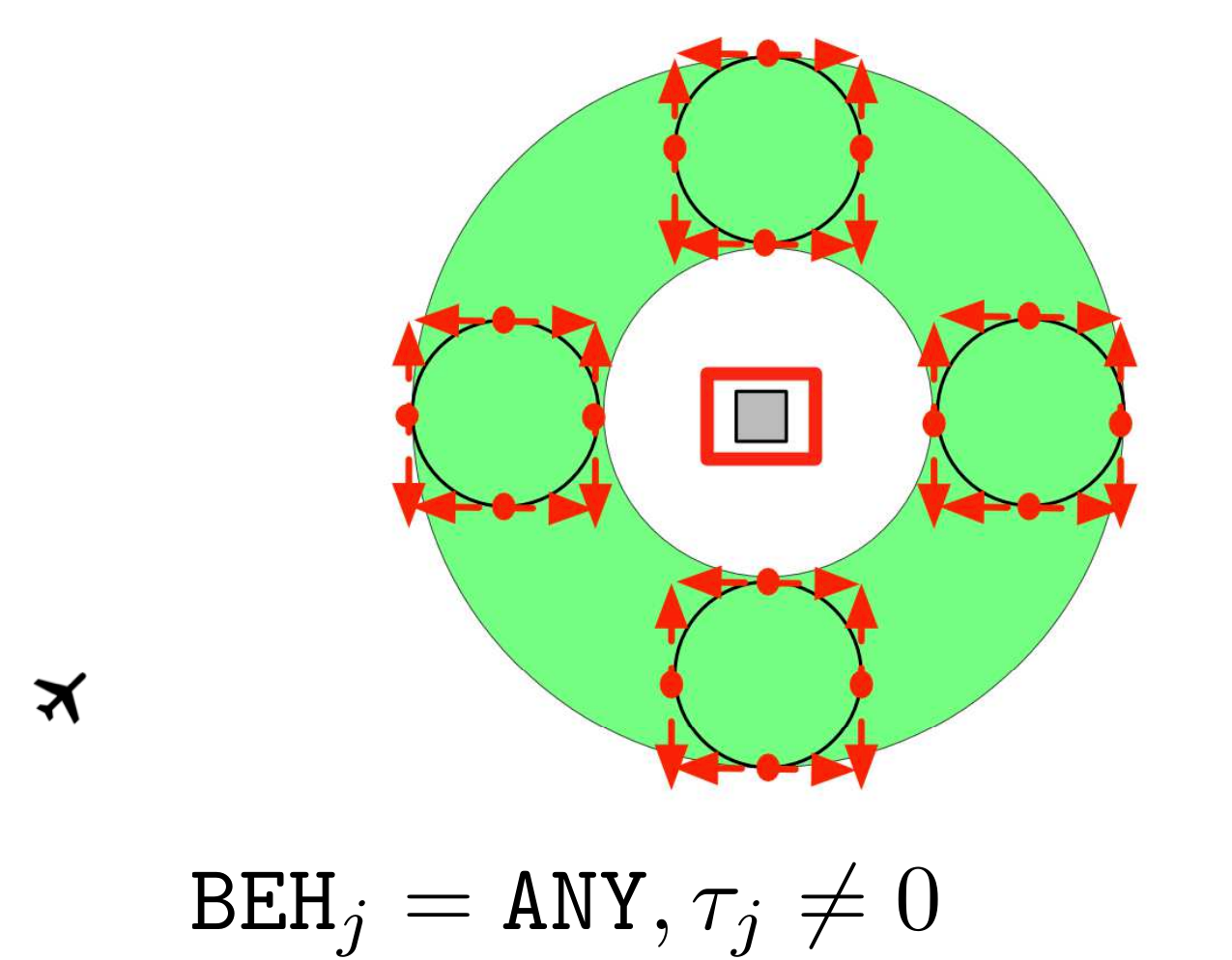}\vrule
\includegraphics[width = 0.19\textwidth]{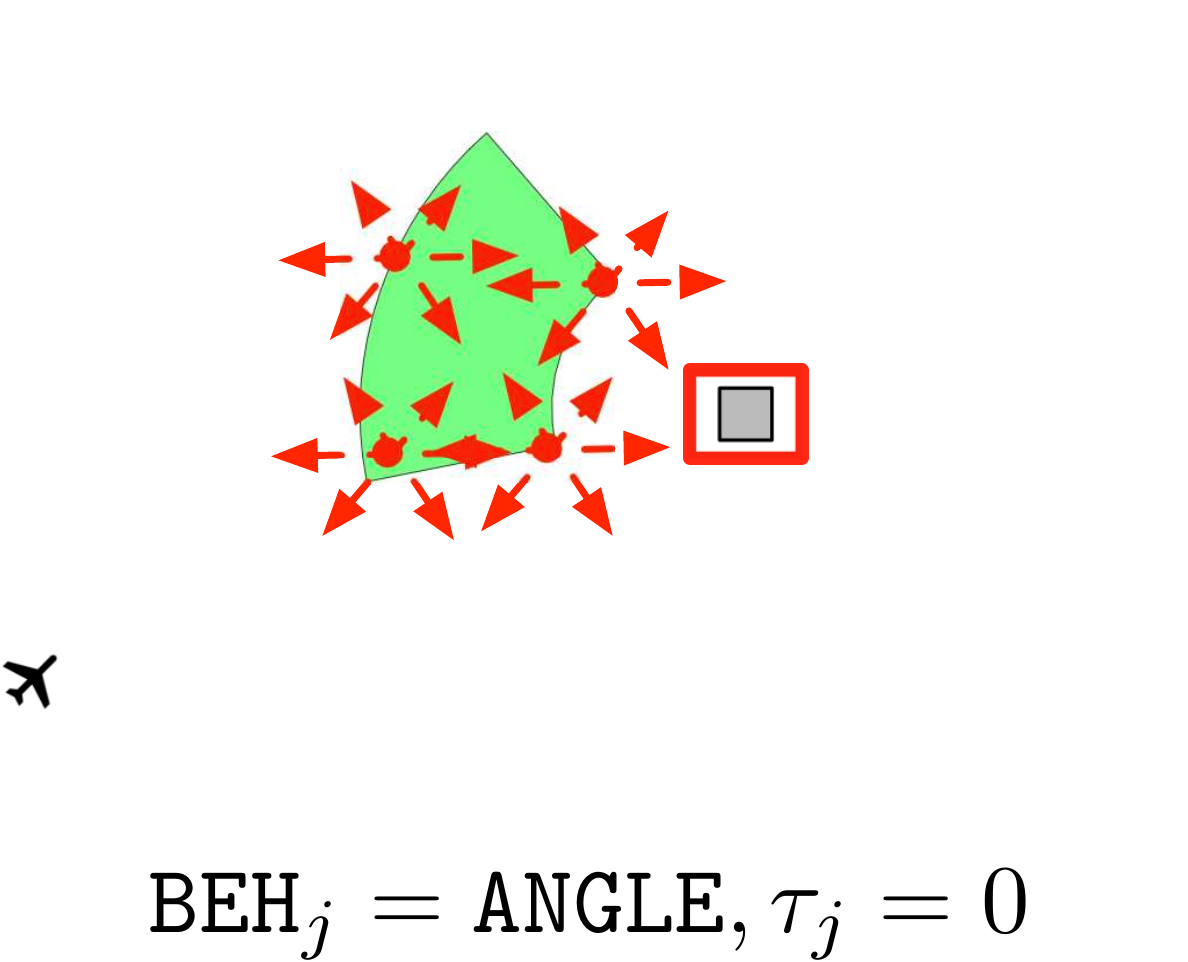}\vrule
\includegraphics[width = 0.19\textwidth]{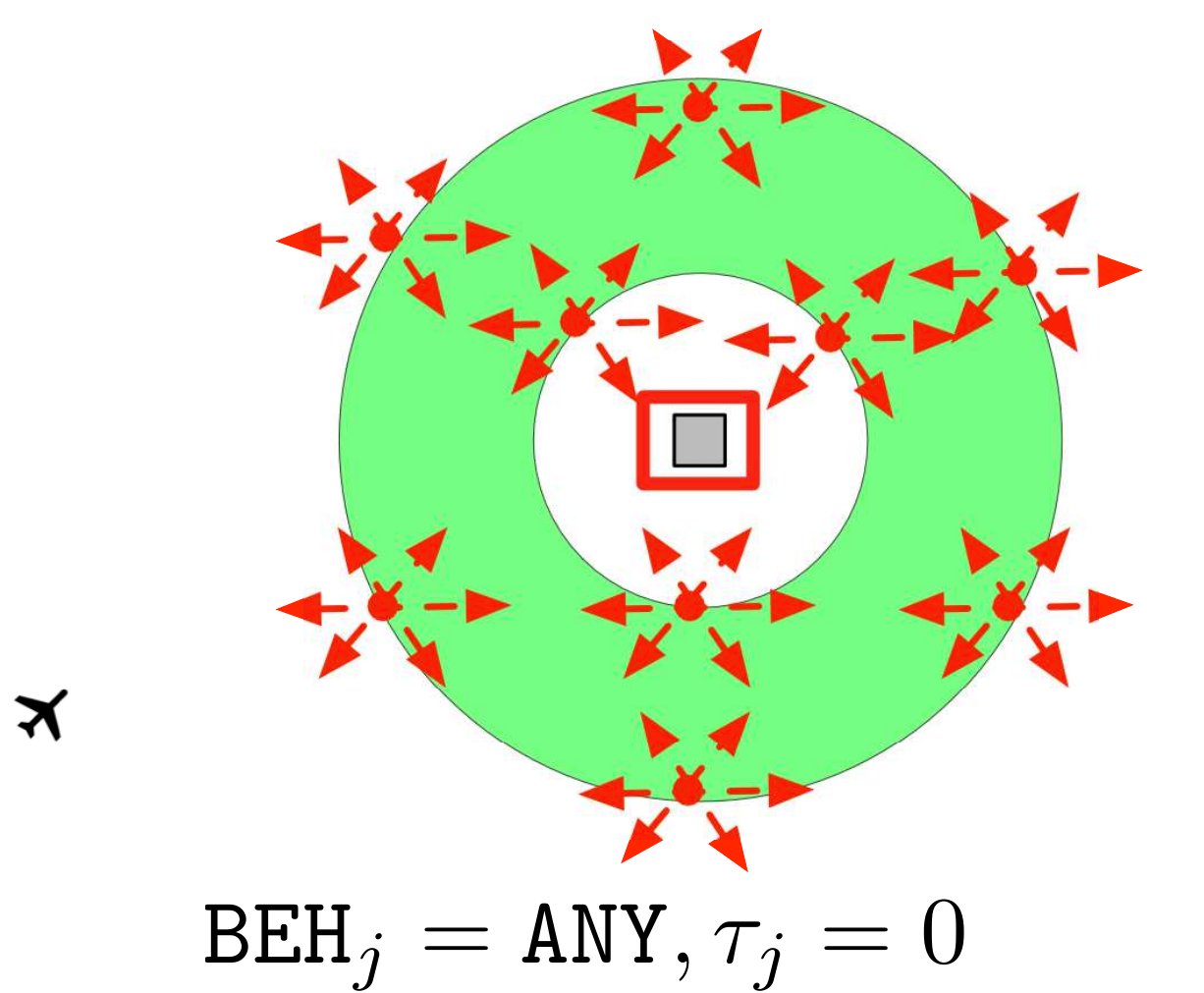}
\caption{Examples of valid configuration samples associated with $T_j$ for various choices of $\param{BEH}_j$ and $\tau_j$.}
\label{fig:example}
\end{figure*}

For the remaining analysis, we assume that imaging parameters are chosen to ensure problem feasibility, i.e., there exists at least $1$ dwell-time maneuver at each target satisfying the aforementioned constraints. 
\begin{remark}[Feasibility]
Feasibility is always ensured if tolerance parameters are sufficiently large. 
\label{rem:feasible}
\end{remark}

\subsection{Problem Statement}
\label{sec:goals}
The goal is to construct an optimal UAV trajectory with the following characteristics: The UAV begins its tour by moving from its initial configuration to a configuration where it can begin imaging a target and, after the initial maneuver, the UAV follows a closed trajectory, along which it images each target according to the required specifications. Note that by separating the initial maneuver from the remaining closed route, we ensure that the target imaging behavior can be effectively repeated if desired. Recall the performance metrics to be minimized: (i) the time required for the UAV to traverse the closed portion of the  generated trajectory (beginning and ending at the first target), and (ii) the time required for the UAV to perform its initial maneuver, i.e., move from $v_0$ to the starting point of the closed portion. Since the metrics are conflicting in general, a tradeoff must be made.
We consider the following formulation of the route-finding problem.
\begin{problem}[Optimal UAV Tour]
Find a UAV tour (consisting of an initial maneuver, and a closed trajectory) that solves the following optimization problem:
\begin{equation}
\begin{split}
\text{Minimize:}\hspace{3mm}&\text{Closed Trajectory Time},\\
\text{Subject To:}\hspace{3mm}& \text{Initial Maneuver Time} \leq \epsilon,\\
&\text{Dynamic Constraints Satisfied (Sec.~\ref{sec:UAV_specs}), and}\\
&\text{Correct Dwell-Time Maneuvers Performed} \\&\hspace{5mm}\text{at Each Target (Sec.~\ref{sec:target}),}
\end{split} 
\label{eqn:multi_opt_1}
\end{equation}
\label{prob:continuous}
where $\epsilon \geq 0$ is a constant parameter. 
\end{problem}
\begin{remark}[Scalarization]
The optimization problem~\eqref{eqn:multi_opt_1} is an $\epsilon$-constraint scalarization of the multi-objective problem. A typical alternative scalarization would instead account for the initial maneuver time within a linear objective: $\alpha \left(\text{Initial Maneuver Time}\right) + \beta \left(\text{Closed Trajectory Time}\right)$ where $\alpha$, $\beta$ are constant parameters. For fixed $\alpha, \beta$, this alternative formulation is equivalent to~\eqref{eqn:multi_opt_1} for some choice of $\epsilon$. Further, parameter selection and subsequent solution of the linear alternative typically requires construction of a Pareto optimal front by solving instances of~\eqref{eqn:multi_opt_1}.
\label{rem:scalarization}
\end{remark}
\begin{remark}[Relation to~\cite{KJO-PO-SD:12}] 
If (i) $\epsilon$ is large (the initial maneuver is inconsequential), and (ii) $\tau_j = 0$ for all $j$ (dwell-times are trivial), then solving Problem~\ref{prob:continuous} is equivalent to solving a \emph{Polygon-Visiting Dubins Traveling Salesperson Problem}, as in~\cite{KJO-PO-SD:12}. Our methods are loosely based on~\cite{KJO-PO-SD:12}, though we consider a more general multi-objective framework that also incorporates non-trivial dwell-times. 
\label{rem:existing}
\end{remark}

\section{Discrete Approximation}
\label{sec:discrete}

Problem~\ref{prob:continuous}, which explicitly considers all target imaging constraints, has an infinite number of potential solutions and is difficult to solve. However, by carefully sampling the UAV configuration space, we can pose a finite, discrete alternative whose optimal solutions approximate those of the original problem. The discrete approximation is still a combinatorial search; however, it is closely related to standard path-finding problems, allowing us to leverage existing solvers to produce high-quality sub-optimal solutions.
This section develops the discrete approximation of interest. 

\subsection{Configuration Space Sampling}
\label{sec:discretize}

Recalling the Dubins model, we sample the UAV configuration space to obtain a finite collection of points of the form $v \coloneqq(x,\theta) \in \real^2\times [0,2\pi)$. These points serve as the basis for the discrete approximation to Problem~\ref{prob:continuous}. Specifically, we choose a set of points that 
each represent the starting and ending configuration of an appropriate dwell-time ``loop'' at some target (valid dwell-time maneuvers each start and end at the same configuration).  That is, each sampled point $v\coloneqq(x,\theta)$ has heading $\theta$ that points in a direction tangent to a valid dwell-time loop (associated with some target $\param{TAR}_v$) passing through the location $x \in \param{VIS}_j$ (Remark~\ref{rem:correspondence}). 
By explicitly pairing each $v$ with its target $\param{TAR}_v$, this procedure creates a natural one-to-one mapping between the generated points and a set of feasible dwell-time maneuvers. As such, subsequent graph formulations can  ``disregard'' dwell-time constraints by using an augmented graph distance. 
Fig.~\ref{fig:example} shows examples of valid sampled sets associated with some $T_j$, for various $\param{BEH}_j$ and $\tau_j$ values. Here, the red dot is the sampled point's location and the arrow represents its heading (distinct points can have the same planar location). 

Algorithm~\ref{alg:discretize} outlines the sampling process. Here, each set $\texttt{DWL}_j$ is defined thusly: If $\tau_j \neq 0$, let $\texttt{DWL}_j$ be the set of points $v \coloneqq (x, \theta) \in \texttt{VIS}_j \times [0,2\pi)$ having location $x$ that lies on the circular image of an appropriate dwell-time maneuver and heading $\theta$ that points in a direction tangent to the same circular image at $x$. If $\tau_j = 0$, let $\texttt{DWL}_j \coloneqq \texttt{VIS}_j \times [0,2\pi)$. Notice Algorithm~\ref{alg:discretize} allows multiple ``copies'' of the same configuration be sampled, provided each is associated with a distinct target, i.e., there may exist $v_{k_1}, v_{k_2} \in V$ with identical locations and headings so long as $\param{TAR}_{v_{k_1}} \neq \param{TAR}_{v_{k_2}}$. 
%
%
\SetAlgorithmName{Algorithm}{ of Algorithms}{Algorithms}
\IncMargin{.0em}
\begin{algorithm}[h]
 {\footnotesize
  \SetKwInOut{Input}{Input}
  \SetKwInOut{Output}{Output}
  \SetKwInOut{Set}{Set}
  \SetKwInOut{Title}{Algorithm}
  \SetKwInOut{Require}{Require}
  \Input{$N \in \mathbb{N}$; $a$; $\texttt{VIS}_j$, $\tau_j$ for all $j \in \until M$} 
  \Output{$V$, $\{\param{TAR}_v\}_{v \in V}$}
  \BlankLine
\nl Initialize $V = \emptyset$\;
\nl    \For{ Each $T_j$}{
\nl 	Construct and parameterize $\texttt{DWL}_j$ by considering the images \\\hspace{3mm} of valid dwell-time maneuvers at target $T_j$\;
\nl		\For{$k \in \until N$}{
\nl 			Sample $v_k\in \texttt{DWL}_j$, associate the target $T_j$ to $v_k$, \\\hspace{3mm}i.e., define $\texttt{TAR}_{v_k} \coloneqq T_j$, and add $v_k$ to $V$\;
		}
  }
  \nl \Return $V$, $\{\param{TAR}_v\}_{v \in V}$
    \caption{\textit{Configuration Space Sampling}}
 \label{alg:discretize}}
\end{algorithm} 
\DecMargin{.0em}
\begin{remark}[Dwell-Time Loops]
It is possible that some $v\in V$ is tangent to multiple, distinct dwell-time loops associated with $\param{TAR}_v$. Notice that all such loops have identical radii, i.e.,  the UAV requires the same amount of time to traverse each. Thus, we can assume without loss of generality that each $v$ is the starting and ending configuration of a single loop associated with $\param{TAR}_v$.
\label{rem:correspondence}
\end{remark}

\subsection{Graph Construction}
\label{sec:graph}

Algorithm~\ref{alg:discretize} returns a discrete set $V$, where each $v\in V$ represents a point in the UAV configuration space $\real^2 \times [0,2\pi)$ that is the starting and ending point of a valid dwell-time maneuver at $\param{TAR}_v$. Recalling that the optimal Dubins path between any two configurations $v_{k_1}, v_{k_2} \in \real^2 \times [0,2\pi)$ is well defined (and easily computed)~\cite{LED:57}, we utilize Algorithm~\ref{alg:graph} to construct a weighted, directed graph $G \coloneqq (V\cup\{v_0\},E,W)$ that effectively discretizes the solution space of Problem~\ref{prob:continuous}. Here, the edge set $E$ contains directed edges connecting each pair of nodes in $V$, along with directed edges connecting the initial UAV configuration $v_0$ with each node in $V$. Weights are defined via an augmented distance that includes both the time required to complete the dwell-time maneuver at the source node and the time required to travel between configurations (recall that optimal Dubins paths are asymmetric in general). We are now ready to formally define the discrete approximation to Problem~\ref{prob:continuous} using the graph $G$. 
\SetAlgorithmName{Algorithm}{ of Algorithms}{Algorithms}
\IncMargin{.0em}
\begin{algorithm}[ht]
 {\footnotesize
  \SetKwInOut{Input}{Input}
  \SetKwInOut{Output}{Output}
  \SetKwInOut{Set}{Set}
  \SetKwInOut{Title}{Algorithm}
  \SetKwInOut{Require}{Require}
  \Input{$N \in \mathbb{N}$;  $V$, $\{\param{TAR}_v\}_{v \in V}$; $v_0$, $s$, $a$, $r$} 
  \Output{$G\coloneqq(V\cup\{v_0\}, E,W)$}
  \BlankLine
  \nl Initialize the edge set $E = \emptyset$\;
\For{ Each pair of distinct points $v_{k_1}, v_{k_1} \in V$}{
\nl Add the directed edges $(v_{k_1}, v_{k_2})$ and $(v_{k_2}, v_{k_1})$ to $E$\;
\nl Set the weight $W\left(v_{k_1}, v_{k_2}\right)$ equal to the sum of: \\\hspace{3mm}(i) \hspace{1mm}the time required to perform the dwell-time \\\hspace{10mm}maneuver associated with $v_{k_1}$, and\\ \hspace{3mm}(ii) the time required to traverse the optimal Dubins \\\hspace{10mm}path from $v_{k_1}$ to $v_{k_2}$\;
\nl Set the weight $W\left(v_{k_2}, v_{k_1}\right)$ equal to the sum of: \\\hspace{3mm}(i) \hspace{1mm}the time required to perform the dwell-time \\\hspace{10mm}maneuver associated with $v_{k_2}$, and\\ \hspace{3mm}(ii) the time required to traverse the optimal Dubins \\\hspace{10mm}path from $v_{k_2}$ to $v_{k_1}$\;
}
\nl Add the initial UAV configuration $v_0$ to the node set of $G$\;
\nl    \For{ Each node $v \in V$}{
\nl Add the directed edges $(v_0, v)$ to $E$\;
\nl Set the weight $W\left(v_0, v\right)$ equal to the time required to \\\hspace{3mm}traverse the optimal Dubins path from $v_0$ to $v$\;
}
\nl \Return $G = (V\cup \{v_0\}, E, W)$
  }
    \caption{\textit{Graph Construction}}
 \label{alg:graph}
\end{algorithm} 
\begin{problem}[Discrete Approximation]
Consider the graph $G\coloneqq(V \cup \{v_0\},E,W)$ resulting from Algorithm~\ref{alg:graph}. Find a sequence $v_1, v_2, \ldots, v_M \in V$ that solves
\begin{equation}
\begin{split}
\text{Minimize:}\hspace{3mm}&W(v_M, v_1) + \sum_{k = 1}^{M-1} W\left(v_k,v_{k+1}\right)\\
\text{Subject To:}\hspace{3mm}&W\left(v_0, v_1\right) \leq \epsilon,\text{ and}\\
&\texttt{TAR}_{v_{k_1}} \neq \texttt{TAR}_{v_{k_2}},\text{ for any } k_1 \neq k_2,
\end{split} 
\label{eqn:multi_opt_graph}
\end{equation}
where $v_0\in \real^2 \times [0,2\pi)$ corresponds to the initial UAV configuration and $\epsilon \geq 0$ is a constant parameter.
\label{prob:graph}
\end{problem}

\section{UAV Tour Construction}
\label{sec:optimal}

Notice that solutions to Problem~\ref{prob:continuous} can be recovered from solutions to Problem~\ref{prob:graph}. Indeed, given a solution $v_1, \ldots, v_M$ to \eqref{eqn:multi_opt_graph}, we recover a feasible solution to~\eqref{eqn:multi_opt_1} by: (i) concatenating the optimal Dubins paths between adjacent nodes in the sequence (appending the path from $v_0$ to $v_1$ at the beginning and the path from $v_M$ with $v_1$ at the end) and (ii) appending the dwell-time trajectory associated to each node $v_1, \ldots, v_M$.
The remainder of our analysis studies the discrete approximation (Problem~\ref{prob:graph}) and its relation to the continuous formulation (Problem~\ref{prob:continuous}).

\subsection{Solving the Discrete Problem}
\label{sec:discrete_solve}
We leverage solutions of a classic graph path-finding problem to construct solutions to~\eqref{eqn:multi_opt_graph}. In particular, we propose a heuristic framework that relates solutions of~\eqref{eqn:multi_opt_graph} to those of an (asymmetric) \emph{Generalized Traveling Salesperson Problem} (GTSP), which is defined for convenience here. 
\begin{problem}[GTSP]
Given a complete, weighted, directed graph $\mathcal{G} \coloneqq (\mathcal{V}, \mathcal{E}, \mathcal{W})$, and a family of finite, non-empty subsets $\{\mathcal{V}_j\subseteq \mathcal{V}\}_{j \in \until M}$, find a minimum weight, closed path that visits exactly one node from each subset $\mathcal{V}_j$.
\label{prob:gtsp}
\end{problem}
\begin{remark}[GTSP Formulation]
A common alternative GTSP formulation requires the closed path to visit \emph{at least} one node from each $\mathcal{V}_j$. If edge weights satisfy a triangle inequality, then this alternative and Problem~\ref{prob:gtsp} are identical. In what follows, we define GTSP instances on a complete subgraph $\mathcal{G}\subseteq G$, where $G$ is the graph constructed in Algorithm~\ref{alg:graph}. In this case, edge weights in $\mathcal{G}$ represent an augmented Dubins distance and, since the Dubins distance function satisfies a triangle inequality~\cite{JTI-JPH:13}, edge weights in $\mathcal{G}$ also satisfy a triangle inequality. Thus, we consider Problem~\ref{prob:gtsp} without loss of generality.
\label{rem:formulation}
\end{remark}

\begin{remark}[GTSP Solutions]
The standard GTSP is NP-hard. However, practical strategies exist for quickly constructing high-quality solutions, e.g., transformation of the problem into a standard ATSP and application of a heuristic solver (see Section~\ref{sec:literature}).

\label{rem:gtsp}
\end{remark}
Note that, in general, Problem~\ref{prob:graph} is \emph{not} equivalent to a GTSP, due to the constraint on the initial maneuver. We can, however, leverage GTSP solution procedures in constructing solutions to the constrained problem. Indeed, a heuristic procedure for constructing solutions to Problem~\ref{prob:graph} using the solutions to related GTSP instances is outlined in Algorithm~\ref{alg:discrete}.
\begin{algorithm}[h]
 {\footnotesize
  \SetKwInOut{Input}{Input}
  \SetKwInOut{Output}{Output}
  \SetKwInOut{Set}{Set}
  \SetKwInOut{Title}{Algorithm}
  \SetKwInOut{Require}{Require}
  \Input{$G = (V\cup\{v_0\},E,W)$, $\{\param{TAR}_v\}$} 
  \Output{$v_1, \ldots, v_M$}
  \BlankLine
  \nl Construct the set $\texttt{INL}_\epsilon \coloneqq \{v \in V\;|\; W(v_0, v) \leq \epsilon\}$\;
  \eIf{$\texttt{INL}_\epsilon$ is empty}{
  \nl \Return ``Problem~\ref{prob:graph} Infeasible''}{
  \nl Select a subset $\texttt{INL}_\epsilon^* \subseteq \texttt{INL}_\epsilon$, whose elements are all \\\hspace{3mm}associated with a single target $T_{\hat{\jmath}}$\;
  \nl Construct the subgraph $\mathcal{G} \coloneqq (\mathcal{V}, \mathcal{E}, \mathcal{W})\subseteq G$ that is \\\hspace{3mm}induced by the node set $\mathcal{V} \subseteq V$, where \\\hspace{3mm}$\mathcal{V} \coloneqq V \backslash \{v\in V\,|\, \texttt{TAR}_v = T_{\hat{\jmath}}, v \notin \texttt{INL}^*_\epsilon\}$\;
  \nl Formulate and solve the GTSP (Problem~\ref{prob:gtsp}) using the graph $\mathcal{G}$ \\\hspace{3mm} and subsets $\mathcal{V}_j \coloneqq \{v \in \mathcal{V}\,|\,\texttt{TAR}_v = T_j\}$\;
  \nl Cyclically permute the GTSP solution to obtain a sequence \\\hspace{3mm}of nodes $v_1, v_2, \ldots, v_M$ with $\texttt{TAR}_{v_1} = T_{\hat{\jmath}}$\;
 \nl \Return $v_1, \ldots, v_M$}}
    \caption{\textit{Heuristic Solution to Problem~\ref{prob:graph}}}
     \label{alg:discrete}
\end{algorithm} 
Here,  $\texttt{INL}_\epsilon$ denotes the set of all nodes in $V$ that can be reached from $v_0$ in time less than $\epsilon$. 
In general, the sequences produced by Algorithm~\ref{alg:discrete} will not be optimal with respect to Problem~\ref{prob:graph}. They will, however, be feasible. Further, if $\texttt{INL}_\epsilon$ has a particular structure, then the subset $\texttt{INL}_\epsilon^*$ (see Algorithm~\ref{alg:discrete}) can be chosen to ensure that the GTSP instance (line $5$) is equivalent to Problem~\ref{prob:graph}. The following results make this discussion precise.
\begin{theorem}[Feasibility]
Algorithm~\ref{alg:discrete} produces a feasible solution to Problem~\ref{prob:graph}.
\label{thm:feasibility}
\end{theorem}
\begin{proof}
The GTSP solution (line $5$) will contain some $v \in \texttt{INL}_\epsilon^*\subseteq \texttt{INL}_\epsilon$. Thus, the permutation operation in line $6$ will produce $v_1, v_2, \ldots, v_M$ with $v_1 \in \texttt{INL}_\epsilon$. It follows readily that the algorithmic output is feasible with respect to Problem~\ref{prob:graph}
\end{proof}
\begin{remark}[Feasibility]
Theorem~\ref{thm:feasibility} does not require construction of an optimal GTSP solution, i.e., the result holds as long as \emph{any} feasible GTSP solution is produced in line $5$.
\label{rem:feasibility}
\end{remark}
\begin{theorem}[Equivalence]
Consider Algorithm~\ref{alg:discrete}. Suppose $\texttt{INL}_\epsilon$ is nonempty and that there exists an index $\hat{\jmath} \in \until M$ satisfying either:
\begin{enumerate}
\item $\texttt{INL}_\epsilon \subseteq \{v \in V\;|\;\texttt{TAR}_v = T_{\hat{\jmath}}\}$, or 
\item $\{v \in V\;|\;\texttt{TAR}_v = T_{\hat{\jmath}}\} \subseteq \texttt{INL}_\epsilon$.
\end{enumerate}
If $\texttt{INL}_\epsilon^*$ (line 3) is chosen as the set of all nodes in $\texttt{INL}_\epsilon$ that are associated with $T_{\hat{\jmath}}$, then optimal solutions of the GTSP in line 5 map to those of Problem~\ref{prob:graph} via the operation in line $6$. That is, if a globally optimal solution to the GTSP in line 5 is produced, then the output $v_1,\ldots, v_M$ of Algorithm~\ref{alg:discrete} is a globally optimal solution to Problem~\ref{prob:graph}. 
\label{thm:gtsp}
\end{theorem}
\begin{proof} 
Any solution $v_1, v_2, \ldots, v_M$ to Problem~\ref{prob:graph} must contain exactly one node associated to each target, where $v_1 \in \texttt{INL}_\epsilon$. If $\texttt{INL}_\epsilon$ only contains nodes associated to $T_{\hat{\jmath}}$ (condition (i)), then no feasible solution to Problem~\ref{prob:graph} contains any $v\notin \texttt{INL}_\epsilon$ with $\texttt{TAR}_v = T_{\hat{\jmath}}$. Thus, there is no loss of generality in considering the GTSP defined over the modified graph $\mathcal{G}$ when $\texttt{INL}_\epsilon^* = \texttt{INL}_\epsilon$ (as opposed to defining a GTSP over the subgraph induced by $V$). The same applies when $\texttt{INL}_\epsilon$ satisfies condition (ii) and $\texttt{INL}_\epsilon^*$ equals the set of all nodes associated with $T_{\hat{\jmath}}$, as this implies $\mathcal{V} = V$. Since cyclic permutation of the node sequence does not affect the cost, any optimal solution to the GTSP in line $5$ can be mapped to an optimal solution of Problem~\ref{prob:graph} via the operation in line $6$. 
\end{proof}
Figure~\ref{fig:thm} shows a graphical illustration of when the theorem conditions are met. We note that the conditions required by Theorem~\ref{thm:gtsp} are typically met whenever target spacing is large. 
\begin{figure}
\centering
\includegraphics[width = 0.45\columnwidth]{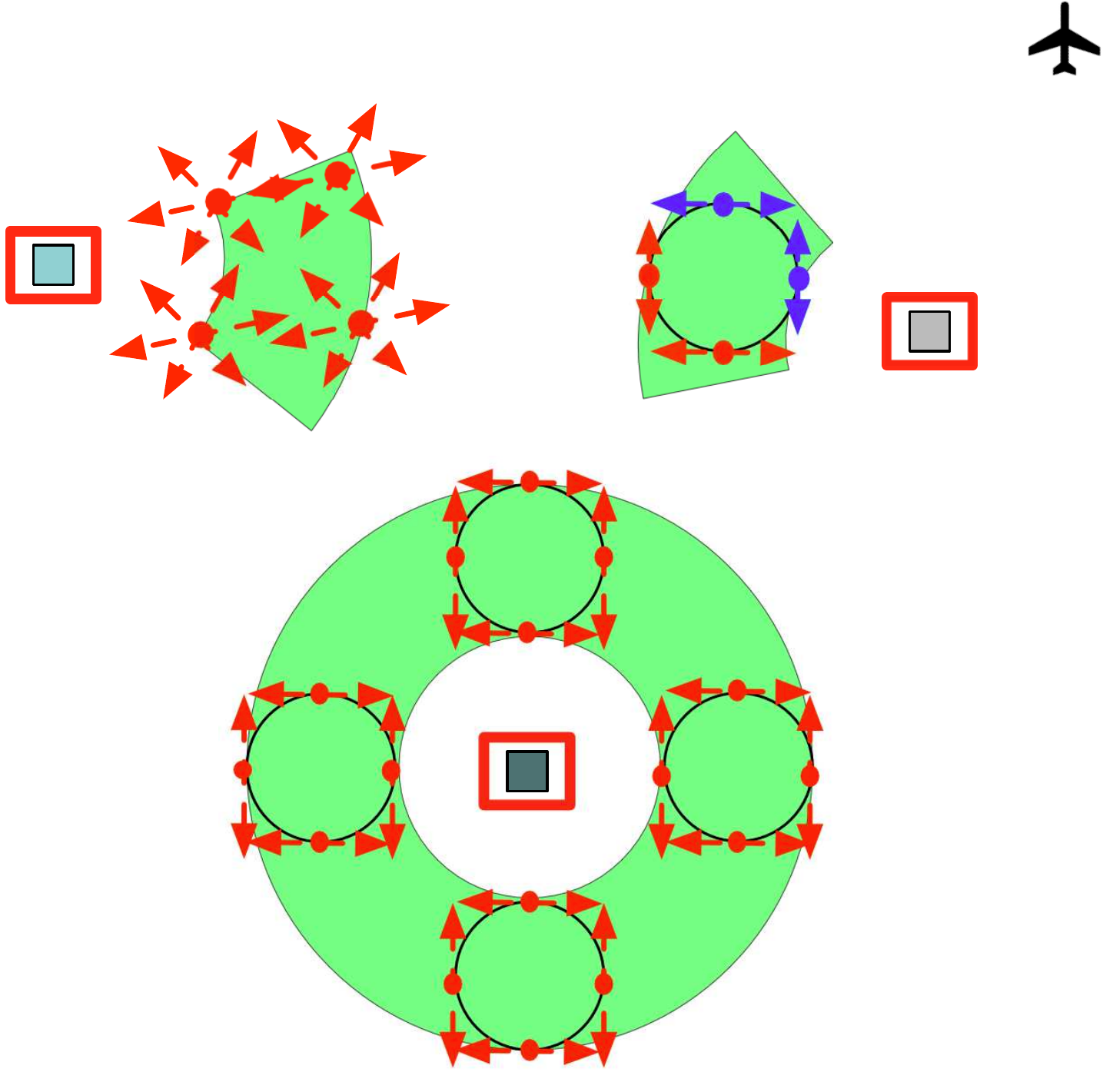}\hspace{3mm}\vrule\hspace{3mm}
\includegraphics[width = 0.45\columnwidth]{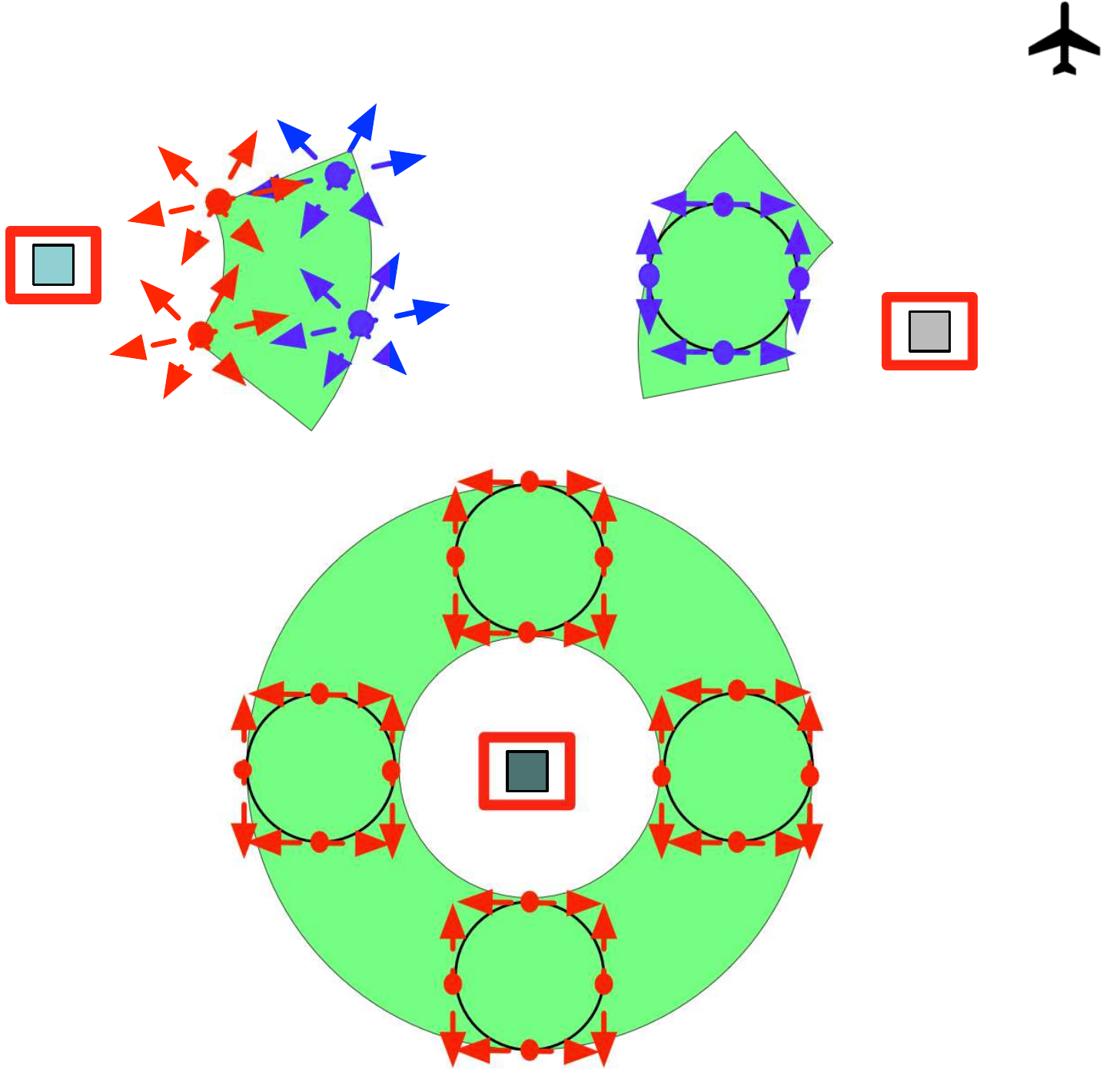}
\caption{Example of when $\texttt{INL}_\epsilon$ (blue nodes) satisfies condition 1 (left) and 2 (right) of Theorem~\ref{thm:gtsp}.}
\label{fig:thm}
\end{figure}

\subsection{Complete Tour Construction}
\label{sec:complete}
Algorithm~\ref{alg:complete} is a heuristic procedure that leverages solutions to the discrete approximation (Problem~\ref{prob:graph}) to  construct solutions to the full routing problem (Problem~\ref{prob:continuous}).
\SetAlgorithmName{Algorithm}{ of Algorithms}{Algorithms}
\begin{algorithm}[ht]
 {\footnotesize
  \SetKwInOut{Input}{Input}
  \SetKwInOut{Output}{Output}
  \SetKwInOut{Set}{Set}
  \SetKwInOut{Title}{Algorithm}
  \SetKwInOut{Require}{Require}
  \Input{$v_0,s,a,r; N \in \mathbb{N}; \{T_j\}_{j \in \until M}\;$} 
  \Output{Complete UAV Route}
  \BlankLine
  \textit{\% Create visibility regions}\;
  \nl  Create target visibility regions via Algorithm~\ref{alg:visibility}.
  \BlankLine
  \textit{\% Create the discrete approximation}\;
  \nl Sample the configuration space and create the graph $G$ via \\\hspace{3mm}Algorithms~\ref{alg:discretize} and~\ref{alg:graph}. Formulate Problem~\ref{prob:graph}\;
  \BlankLine
  \textit{\% Solve the discrete approximation}\;
  \nl Construct a solution $v_1, \ldots, v_M$ to Problem~\ref{prob:graph} via Algorithm~\ref{alg:discrete}\;
  \If{Algorithm~\ref{alg:discrete} returns an error (Problem~\ref{prob:graph} is infeasible)}
  {\nl \Return ``Error: Discrete Approximation Infeasible''}
  \BlankLine
  \textit{\% Convert the solution of Problem~\ref{prob:graph} into a solution to Problem~\ref{prob:continuous}}\;
  \nl Construct the optimal Dubins path that visits the nodes in the \\\hspace{3mm}following order: $v_0, v_1, \ldots, v_M, v_1$\; 
  \nl Append dwell-time maneuvers to recover a solution to Problem~\ref{prob:continuous}.\;
  \nl \Return UAV Route: Initial Maneuver + Closed Trajectory
  }
    \caption{\textit{Heuristic Tour Construction via GTSPs}}
 \label{alg:complete}
\end{algorithm} 

Solutions produced by Algorithm~\ref{alg:complete} are not optimal in general, though they will exhibit structural characteristics that generally improve in quality (with respect to Problem~\ref{prob:continuous}) as the sampling granularity is made increasingly fine. Indeed, Algorithm~\ref{alg:complete} is a \emph{resolution complete} in some sense, providing justification for the sampling-based approximation approach. A precise characterization of the resolution completeness properties is included as an Appendix.

\section{Numerical Examples}
\label{sec:example}
We illustrate our algorithms through numerical examples. For each simulated mission, solutions to Problem~\ref{prob:continuous} are constructed via Algorithm~\ref{alg:complete}, where GTSPs are solved through transformation into an equivalent ATSP (see~\cite{CEN-JCB:91a}) that is subsequently solved using the Lin-Kernighan heuristic (as implemented by LKH~\cite{KH:00}). In all cases, the set $\param{INL}_\epsilon^*$ (Algorithm~\ref{alg:discrete}) is chosen as the set of all points in $\param{INL}_\epsilon$ associated with some single target (satisfying Theorem~\ref{thm:gtsp} conditions whenever possible). A slightly modified version of Algorithm~\ref{alg:discretize} is used for sampling in which the number of samples, $N$, associated with each target is not fixed \emph{a priori}, but instead is determined by creating a grid of samples within the appropriate sampling subsets. The grid spacing is determined by $3$ parameters $\delta r$, $\delta \theta$, and $\delta \alpha$, which represent, loosely, the radial location spacing, the angular location spacing, and the angular heading spacing, resp. The parameters $\delta r$, $\delta \theta$, and $\delta \alpha$ are inversely proportional to the number of samples at each target, and the sampled set is dense in the limit as spacing parameters jointly tend to $0$.

\subsection{Pareto-Optimal Front}
\label{sec:first_example}
The first example is a $5$ target mission with the following UAV parameters: $r = 750$ m, $a = 1000$ m, $s = 39$ m/s, and $v_0 = ((-2500,500)\text{ m},0) \in \mathbb{R}^2\times [0,2\pi)$. Target parameters are shown in Table~\ref{tab:target_data} (tolerances are large enough for feasibility).
\setlength{\tabcolsep}{8pt}
\begin{table*}
\caption{Target Input Data}
\centering
\begin{tabular}
{cccccc}
\midrule\addlinespace[1pt]\midrule $T_j$& $t_j$ & $\mathtt{Beh}_j$ & $\tau_j$& $\left[\supscr{\phi_j}{A}-\supscr{\Delta_j}{A},\supscr{\phi_j}{A}+\supscr{\Delta_j}{A}\right]$ & $ \left[\supscr{\phi_j}{T}- \supscr{\Delta_j}{T},\supscr{\phi_j}{T}+\supscr{\Delta_j}{T}\right]$ \\
\midrule
$T_1$ & $(5000, -5000)\text{ m}$ & $\param{FULL}$ &$2$\hspace{2mm}& $-$ & $[\frac{\pi}{8}, \frac{3\pi}{8}]$ \\
$T_2$ & $(4300, -1750)\text{ m}$ &$\param{ANGLE}$  & $1$ & $[\frac{\pi}{4}, \frac{3\pi}{4}]$ & $[\frac{\pi}{8}, \frac{3\pi}{8}]$ \\
$T_3$ & $(0, 4000)\text{ m}$ &$\param{FULL}$  & $3$ & - & $[\frac{\pi}{8}, \frac{3\pi}{8}]$ \\
$T_4$ & $(-8000, -2000)\text{ m}$ & $\param{ANY}$ &$1$ & - & $[\frac{\pi}{8}, \frac{3\pi}{8}]$ \\
$T_5$ & $(-2000, 8000)\text{ m}$ &$\param{ANGLE}$  &$0$ & $[\frac{3\pi}{2}, 2\pi]$& $[\frac{\pi}{8}, \frac{3\pi}{8}]$
\end{tabular}
\label{tab:target_data}
\end{table*}
The approximate Pareto-optimal front (with respect to $\epsilon$) for Problem~\ref{prob:continuous} as a function of the sampling granularity is shown in Figure~\ref{fig:pareto}. 
The figure also shows illustrations of solutions produced at spacing condition 5 when $\epsilon = 65$ s (left) and $\epsilon = 205$ s (right). Recalling that (i) Theorem~\ref{thm:gtsp} does not hold for all $\epsilon$, and (ii) heuristic solvers for GTSPs do not guarantee global optima, to obtain the best approximation of the Pareto-optimal front, the following steps were taken to generate each curve: First, Algorithm~\ref{alg:complete} was called for a series of $\epsilon$ values, and the resulting initial maneuver/closed trajectory times were recorded. Then, in post-processing, the approximate Pareto-optimal curve was generated by selecting, for each $\epsilon$, the lowest cost route satisfying the initial maneuver constraint out of the solutions produced in the computation stage.  Note that increasing $\epsilon$ corresponds to relaxing the initial maneuver constraint, and thus the cost is non-increasing in $\epsilon$. Notice also that the Pareto-optimal fronts shift toward zero as the sampling spacing is decreased. 
\begin{table}
\caption{Spacing Conditions}
\begin{tabular}
{cccc}
\midrule\addlinespace[1pt]\midrule
Spacing Condition & $\delta r$& $\delta \theta$& $\delta \alpha$\\
\hline
1 & $1000 \text{ m}$& $\pi$ & $\pi$ \\
2 &  $500 \text{ m}$&  $\pi$ & $\pi$ \\
3 & $500 \text{ m}$&  $\pi/2$ & $\pi/2$ \\
4 & $250 \text{ m}$&  $\pi/2$ & $\pi/2$\\
5 & $250 \text{ m}$&  $\pi/4$ & $\pi/4$\\
6 & $125 \text{ m}$& $\pi/4$ & $\pi/4$\\
7 & $125 \text{ m}$& $\pi/8$ & $\pi/8$
\end{tabular}
\label{tab:pareto}
\end{table}
\begin{figure*}
\centering
\includegraphics[width = .7\columnwidth]{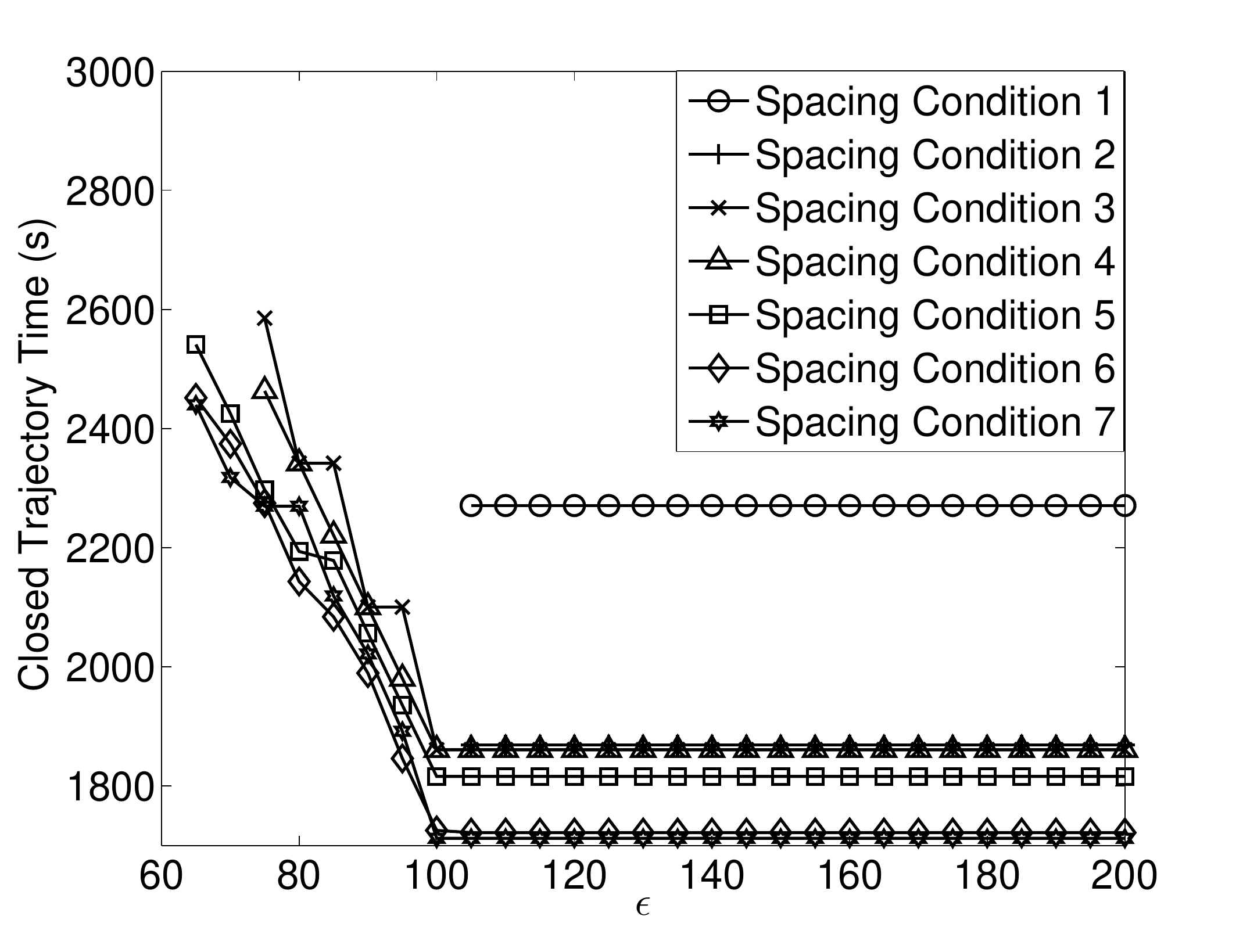}
\includegraphics[width = .7\columnwidth]{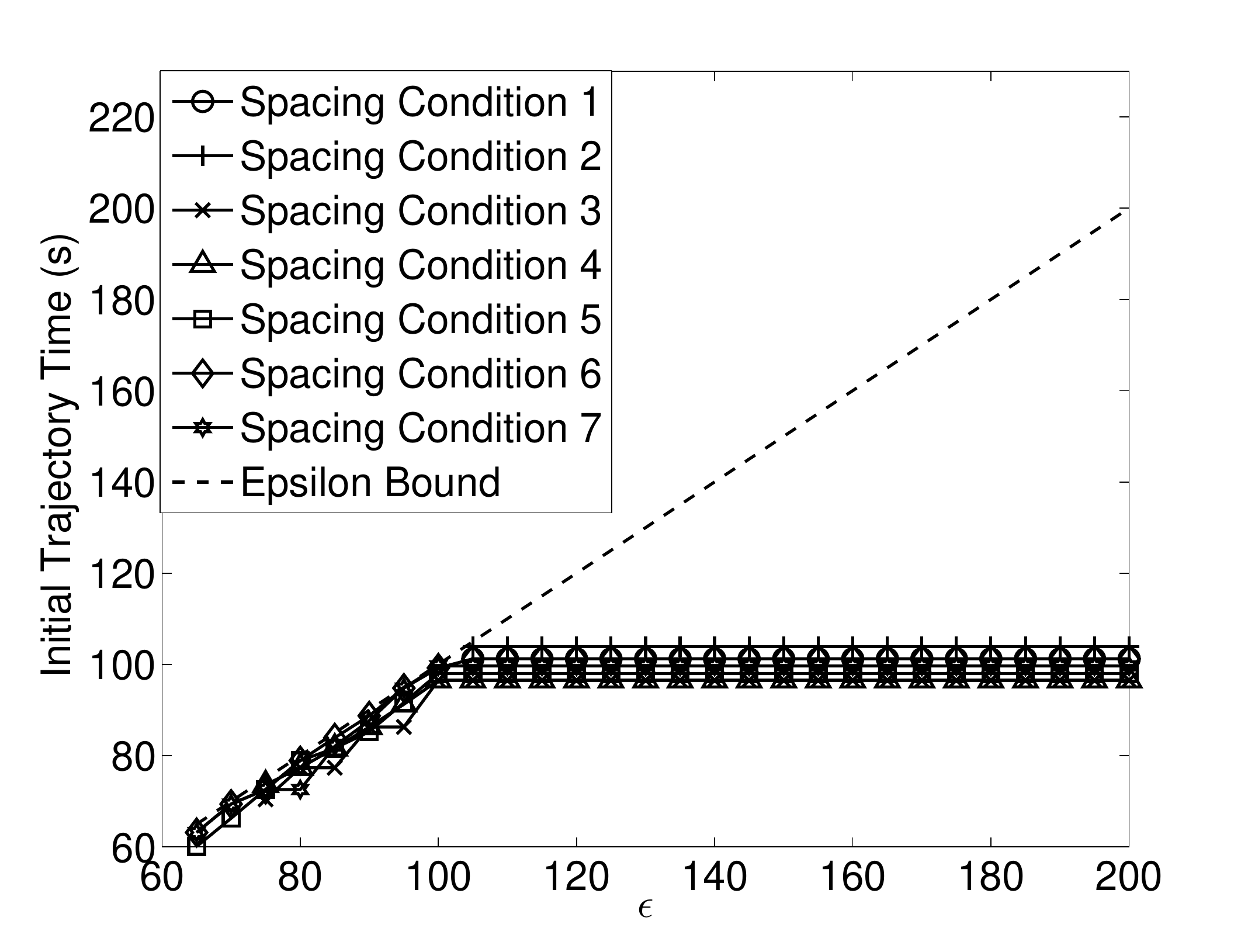}\\
\hspace{-5mm}
\includegraphics[width = .62\columnwidth,height = 45mm]{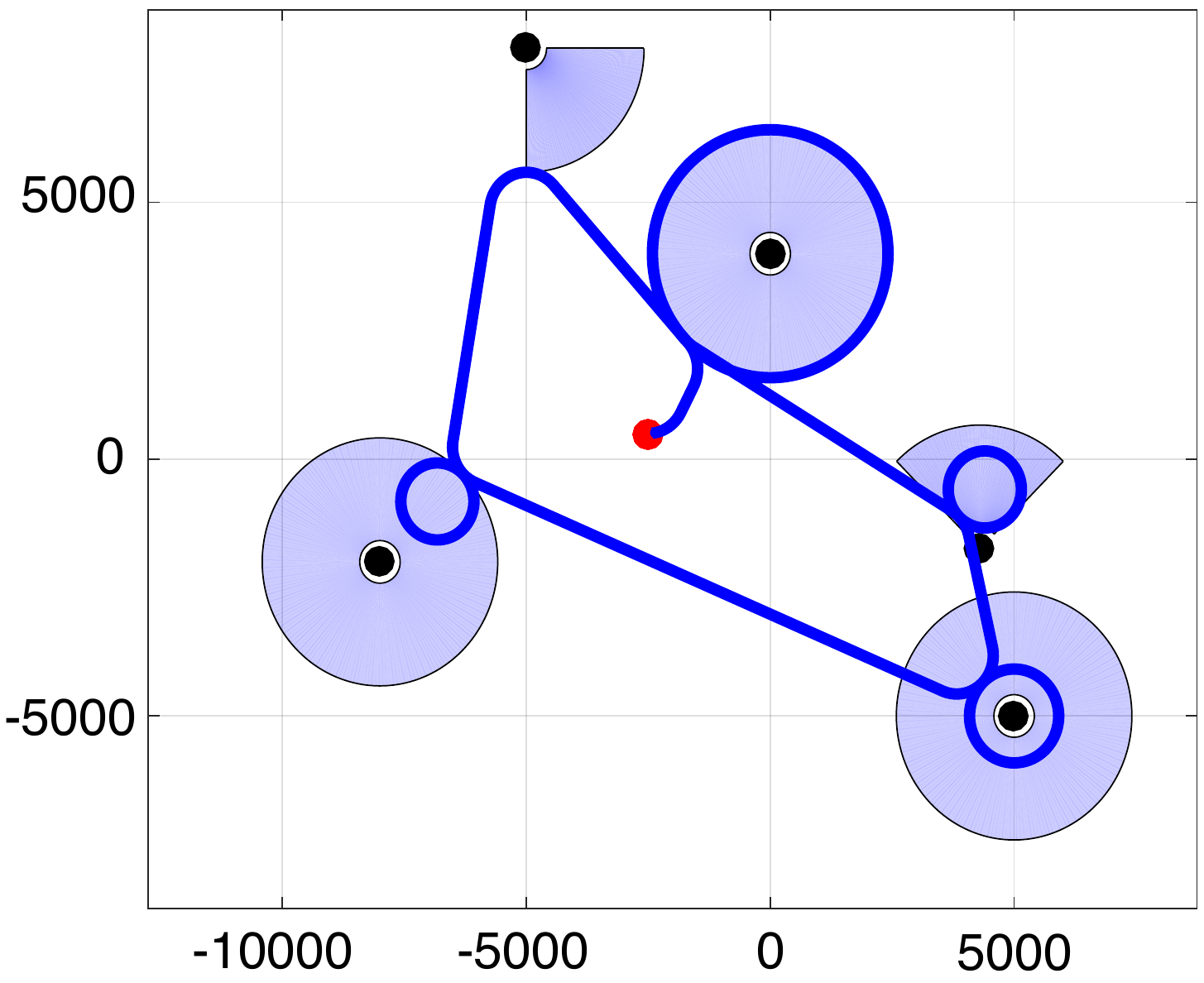}\hspace{7mm}
\includegraphics[width = .62\columnwidth, height = 45mm]{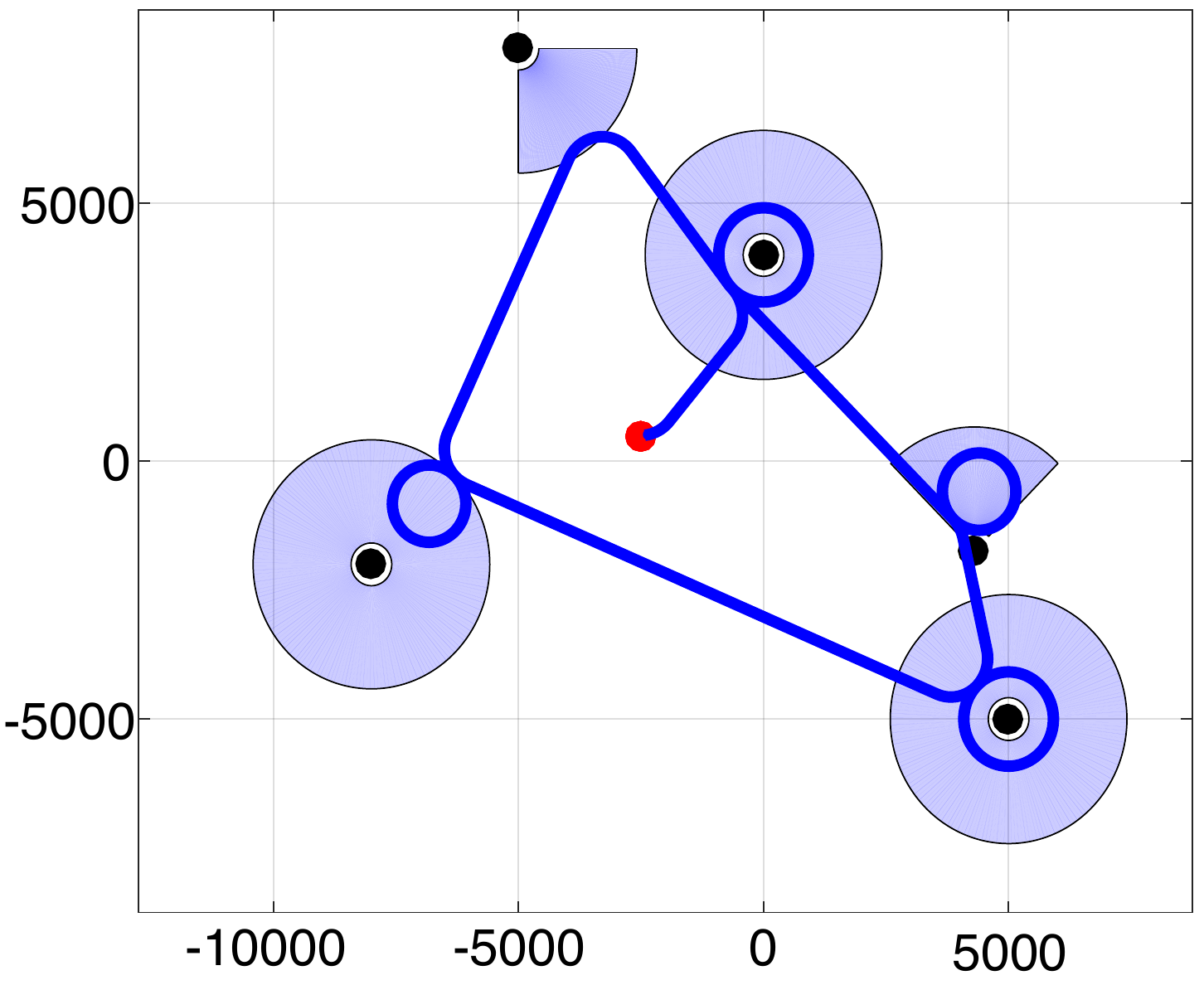}
\caption{Approximate Pareto-Optimal Front and Example Routes for the Spacing Conditions in Table~\ref{tab:pareto}}
\label{fig:pareto}
\end{figure*}
\begin{remark}[Pareto-optimal Fronts]
The optimal cost value in Problem~\ref{prob:continuous} is only sensitive to changes in $\epsilon$ over some finite set $\bigcup_{j = 1}^M [\underline{\epsilon}_j, \overline{\epsilon}_j]$ that depends primarily on the target and UAV locations. This structure can be exploited for more efficient computation of Pareto-optimal fronts.
\end{remark}

\subsection{Performance}
\label{sec:performance}
The next example illustrates the performance of Algorithm~\ref{alg:complete} in comparison to an incremental, ``greedy'' alternative that operates as follows: Visibility region creation and configuration space sampling are done using Algorithms~\ref{alg:visibility} and ~\ref{alg:discretize}. Starting with the initial UAV configuration, each successive UAV destination is chosen by selecting the closest node (Dubins distance) associated with a target that has not yet been imaged. A valid route is constructed by appending dwell-time maneuvers and connecting the last selected configuration ($v_M$) with the first ($v_1$). We consider a $5$ target mission with the same UAV parameters and the same target locations, imaging behaviors, and tolerances as in the previous example. However, we vary the number of dwell-time loops associated with each target (assume each target requires same number of loops).

The difference between the closed trajectory times produced by the greedy method and those produced by Algorithm~\ref{alg:complete} as a function of $\epsilon$ under the spacing condition $5$ (Fig.~\ref{fig:pareto}) is shown in Fig.~\ref{fig:performance}. Notice that the relative performance of Algorithm~\ref{alg:complete} improves as both $\epsilon$ and the number of dwell-time loops at each target are increased. In this example, the performance of the greedy search method can be made arbitrarily poor by increasing the number of dwell-time loops. This result is primarily due to targets that require a 360-degree view, since the greedy heuristic generally chooses those points lying on the perimeter of the visibility region, which are very far from the target location. Thus, increasing the number of dwell-time loops (performed at constant radius) can dramatically increase total tour times. As such, Algorithm~\ref{alg:complete} can provide a significant advantage over similar incremental planning strategies when non-trivial dwell-times are required.
\begin{figure}
\centering
\includegraphics[width = .7\columnwidth]{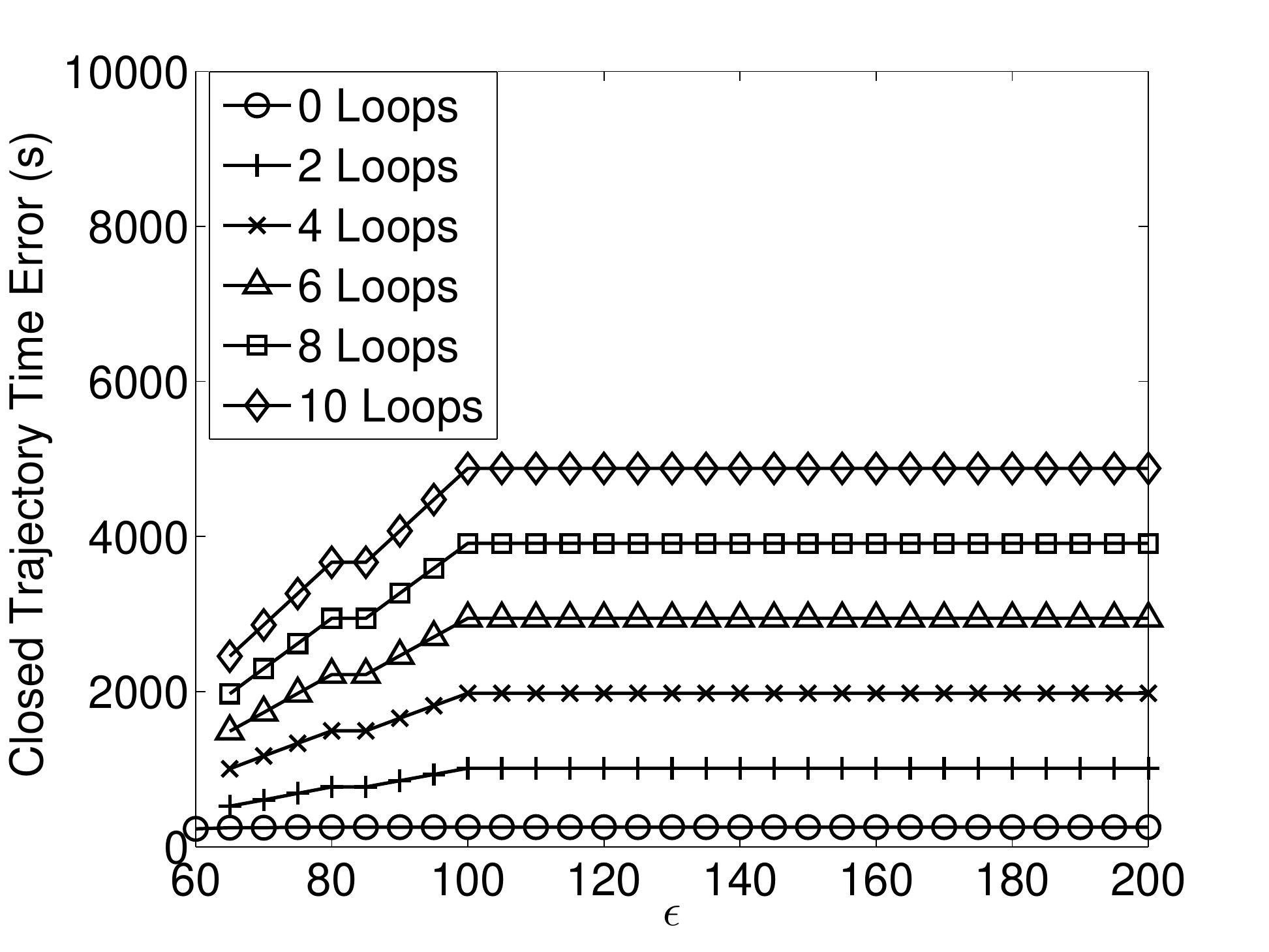}
\caption{Relative performance of the greedy algorithm.}
\label{fig:performance}
\end{figure}

\subsection{Resolution Completeness}
\label{sec:resolution}
To illustrate the resolution completeness properties, consider a mission with $2$ targets whose relevant data is in Table~\ref{tab:target_data_2}. The UAV has the same altitude, velocity, and minimum turning radius as the UAVs in the previous examples, but has initial configuration $v_0 = ((0,0)\text{ m},\pi/7) \in \real^2 \times [0,2\pi)$.
\begin{table*}
\caption{Target Input Data}
\centering
\begin{tabular}
{cccccc}
\midrule\addlinespace[1pt]\midrule $T_j$& $t_j$ & $\mathtt{Beh}_j$ & $\tau_j$& $\left[\supscr{\phi_j}{A}-\supscr{\Delta_j}{A},\supscr{\phi_j}{A}+\supscr{\Delta_j}{A}\right]$ & $ \left[\supscr{\phi_j}{T}- \supscr{\Delta_j}{T},\supscr{\phi_j}{T}+\supscr{\Delta_j}{T}\right]$ \\
\midrule
$T_1$ & $(2131.8,1026.7)\text{ m}$ & $\param{ANY}$ & $0$ & $-$ & $[\frac{\pi}{6}, \frac{\pi}{3}]$ \\
$T_2$ & $(-13840, -5833)\text{ m}$ &$\param{ANY}$  & $1$ & $-$ & $[\frac{\pi}{8}, \frac{3\pi}{8}]$ 
\end{tabular}
\label{tab:target_data_2}
\end{table*}
This example has been carefully constructed so that the optimal solution is easily deduced for select $\epsilon$ conditions. In particular, when $\epsilon\geq 25$ s, the optimal tour involves the UAV making an immediate left turn at minimum radius, and then visiting the remaining target for a total closed tour length of $848.62$ s. When $\epsilon = 16.26$ s, the UAV is constrained to make a shorter initial maneuver, so the UAV travels straight until it hits the visibility region, and then proceeds with the remainder of the tour for an optimal closed tour length of $881.14$ s. A schematic showing optimal routes for the two $\epsilon$ conditions just described is shown in the top left portion of Fig.~\ref{fig:optimal}. (the red curve corresponds to the initial turn for the $\epsilon = 16.26$ s case; the remainder of the route is identical). When $\epsilon\geq 25$ s, the problem instance is non-degenerate (Definition~\ref{def:degeneracy_closed}) and satisfies the conditions required for resolution completeness (Theorem~\ref{thm:resolution}). Therefore, if a reasonable sampling method and GTSP solver are used, the solutions produced by Algorithm~\ref{alg:complete} will tend toward the global optimum with finer sampling. This behavior is illustrated by the top right plot in Figure~\ref{fig:optimal}, which shows the relative error between the cost produced via Algorithm~\ref{alg:complete} and the globally optimal cost (difference divided by the optimum) for the sampling conditions listed in the table when $\epsilon = 130$ s. Notice that the relative error tends to zero with finer discretization. Also note that, for this example, the optimal solution to the discrete approximation involved vertices located on visibility region boundaries, resulting in an insensitivity to radial grid spacing.
In contrast, when $\epsilon = 16.26$ s, the problem becomes degenerate, since there is a single configuration to which the UAV can travel in order to satisfy the $\epsilon$ bound. In this case, Problem~\ref{prob:graph} is infeasible for any sampling scheme that does not choose the precise configuration in question. Therefore, Algorithm~\ref{alg:complete} is not resolution complete in the sense of Theorem~\ref{thm:resolution} when $\epsilon = 16.26$. Note, however, that resolution completeness holds if $\epsilon$ is increased by any arbitrarily small positive amount.
\begin{figure}
\centering
\includegraphics[width = 0.65\columnwidth, height = 45mm]{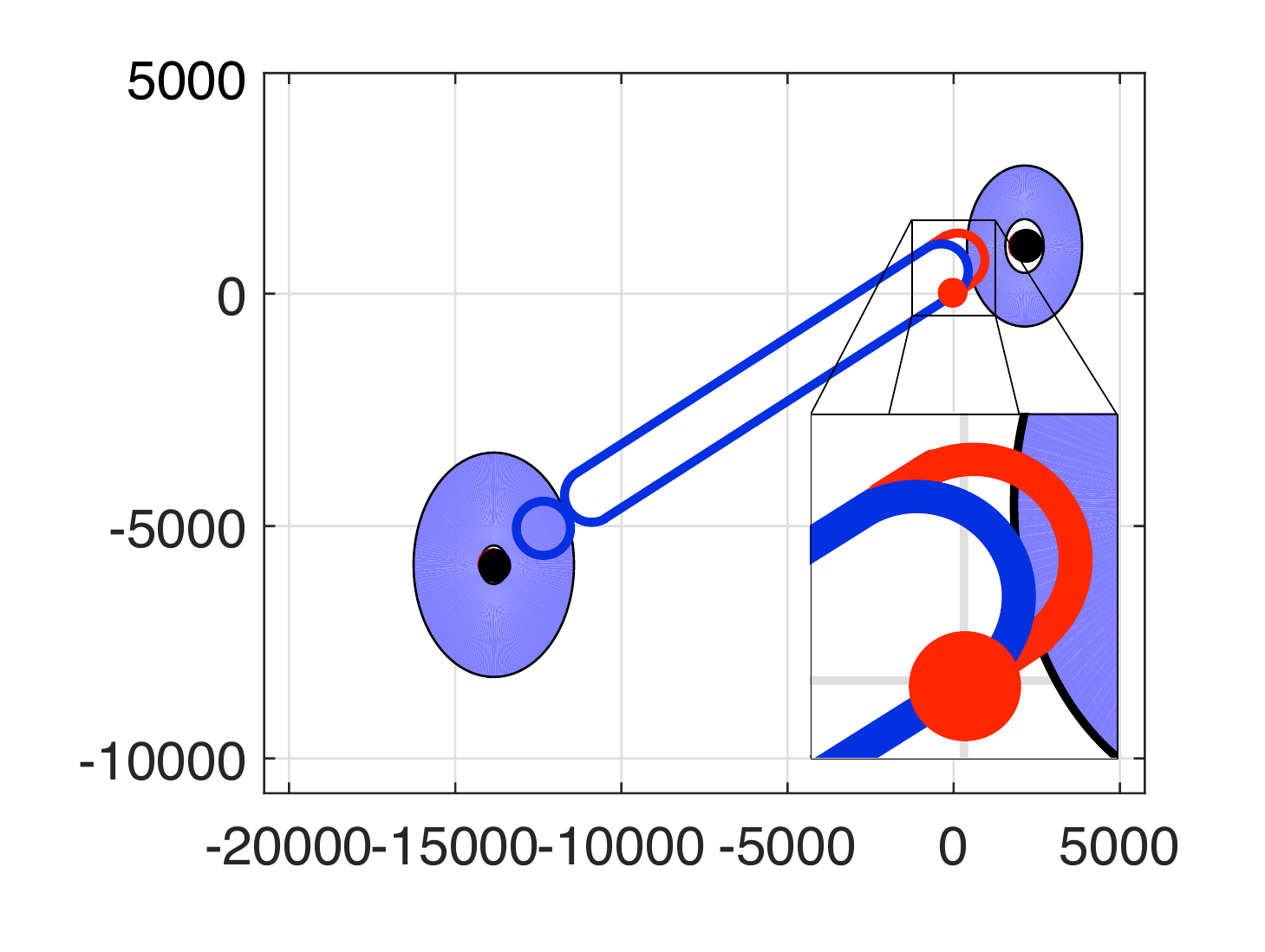}\hspace{5mm}
\caption{Optimal routes when $\epsilon = 16.26$ and $25$ s.}
\label{fig:optimal}
\end{figure}
\begin{figure}
\centering
\includegraphics[width = 0.65\columnwidth]{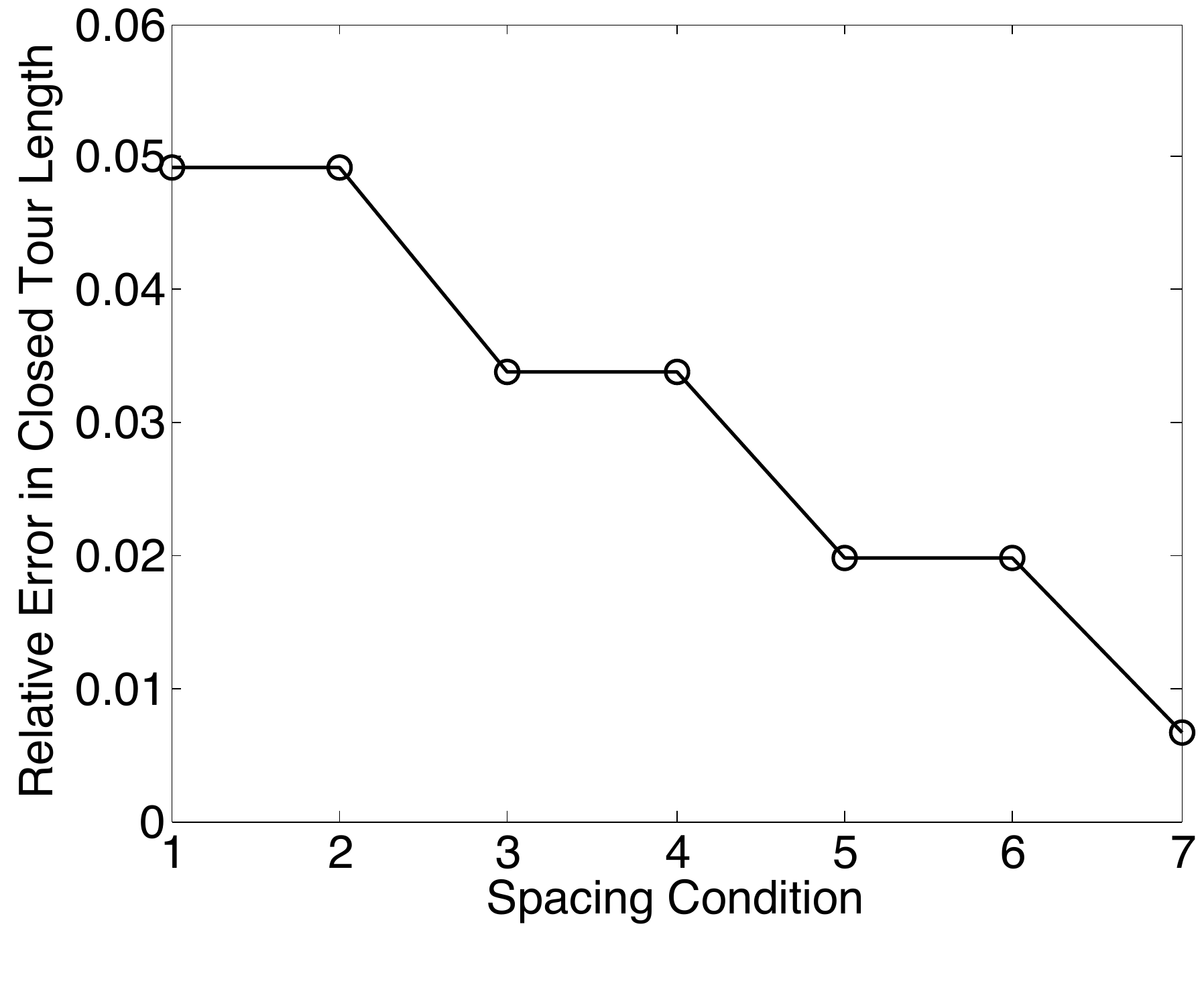}
\caption{Relative cost error when $\epsilon = 130$ s for the spacing conditions in Table~\ref{tab:pareto}.}
\label{fig:optimal}
\end{figure}

\section{Conclusions}
\label{sec:conclusions}

A novel algorithmic framework was presented for constructing unmanned aerial vehicle trajectories for surveillance of multiple targets. Adopting reasonable constraints on optimal dwell-time behavior and UAV maneuvers, the framework works to balance mission goals, namely, the closed trajectory and the initial maneuver times, while simultaneously accommodating both viewing and dwell-time constraints associated with each target. In particular, an $\epsilon$-constraint scalarization method is applied to rigorously pose the multi-objective problem as a constrained optimization, which is then approximated by a finite, path-planning problem that results from careful sampling of the UAV configuration space. Solutions to related GTSP problems can then be utilized to construct solutions to the discrete approximation which, in many cases, map directly to globally optimal solutions of the discrete problem. These solutions were then mapped to solutions of the continuous problem. It was shown that, under certain conditions, the complete heuristic procedure is resolution complete in the sense that solution quality generally increases with increasingly fine sampling.

Avenues of future research include the expansion to the multi-vehicle case, explicit comparisons with other routing schemes (e.g., Markov chain-based schemes), and an investigation of alternative discretization strategies. In addition, incorporation of uncertain dwell-times and the explicit pairing with other facets of complex missions, e.g. operator analysis of imagery, should be explored.

\appendix
\section*{Appendix: Resolution Completeness of Algorithm~\ref{alg:complete}}
We start with some preliminary notions. For this appendix, we assume that all angles $\theta \in [0,2\pi)$ are equivalently represented as points on the unit circle $S^1 \coloneqq \{x \in \real^2\;|\; \|x\| = 1\}$ via the relation $\theta \mapsto (\cos(\theta),\sin(\theta))$. As such, we assume without loss of generality that the UAV configuration space is $\real^2 \times S^1$. Recall that $\texttt{DWL}_j$ is defined as the set of configurations from which a UAV can begin executing a feasible dwell-time maneuver at $T_j$. We first parameterize the set of UAV routes that can be produced by Algorithm~\ref{alg:complete}. 
\begin{definition}[Initial Maneuver Set]
The parametrized \emph{initial maneuver set}, $\texttt{INL}_\epsilon$, is defined \begin{equation*}\texttt{INL}_\epsilon\coloneqq \left.\left\{v \in \bigcup_{j = 1}^M \texttt{DWL}_j\;\right|\; \texttt{DIST}(v_0,v) \leq \epsilon\right\},\end{equation*}
where $\texttt{DIST}(v_0,v)$ is the time required for the UAV to traverse the optimal Dubins path from $v_0$ to $v$.
\label{def:initial}
\end{definition}
Note that we have re-defined $\texttt{INL}_\epsilon$ as the continuous analog of the discrete set of Section~\ref{sec:discrete}.
\begin{definition}[Closed Trajectory Set]
The parametrized \emph{closed trajectory set}, $\texttt{CLS}_\epsilon$, is defined 
\[
\cramped{\texttt{CLS}_\epsilon\coloneqq  \bigcup_{\mathclap{\sigma \in \texttt{Perm}(M)}}  (\texttt{DWL}_{\sigma(1)}\cap\texttt{INL}_\epsilon)\times \texttt{DWL}_{\sigma(2)}\times\cdots\times \texttt{DWL}}_{\sigma(M)},\]
where $\texttt{Perm}(M)$ is the set of all permutations of $\until M$.
\label{def:param}
\end{definition}
Notice that we can map each $\textbf{v}\coloneqq (v_1, v_2, \ldots, v_M) \in \texttt{CLS}_\epsilon$ to a feasible solution of Problem~\ref{prob:continuous} by (i) appending the initial configuration $v_0$, (ii) constructing optimal Dubins routes between successive nodes (connecting $v_M$ to $v_1$), and (iii) appending dwell-time maneuvers.
Strictly speaking, such a mapping  is not unique, since, for some $j$, there may exist $v \in \texttt{DWL}_j$ that is the starting configuration for multiple distinct dwell-time maneuvers (associated to the same target). However, recalling Remark~\ref{rem:correspondence}, this ambiguity is inconsequential to the present analysis and we can assume without loss of generality that, for each $j$, the correspondence between $\param{DWL}_j$ and the set of dwell-time maneuvers at $T_j$ is bijective. Under this assumption, any optimal solution to Problem~\ref{prob:continuous} is uniquely associated with some element in the set $\texttt{CLS}_\epsilon$ via the mapping just described.
With this in hand, we now assign an appropriate cost to each element of $\texttt{CLS}_\epsilon$.
\begin{definition}[Cost]
Define the map $\map{\texttt{LGTH}}{\texttt{CLS}_\epsilon}{\real_{\geq}}$ where $\texttt{LGTH}((v_1, v_2, \ldots, v_M))$ is the time required for the UAV to traverse the closed trajectory (beginning at $v_1$) that sequentially touches and performs the dwell-time maneuver associated with all vertices $v_1, \ldots, v_M$.
\label{def:cost}
\end{definition}
For each $T_j$, the length of an appropriate dwell-time maneuver varies continuously as a function of the maneuver's starting configuration (element in $\texttt{DWL}_j$). As such, results in~\cite{KJO-PO-SD:12} imply that (i) the map $\texttt{LGTH}$ is well-defined and continuous except on a finite set of $(3M-1)$-dimensional smooth surfaces embedded in $(\real^2\times S^1)^M$, and (ii) in the limit from one side of each discontinuity surface, $\texttt{LGTH}$ is continuous up to and on the surface. The presence of discontinuity sets necessitates the following definition. Here, $\texttt{LGTH}^*\coloneqq \min \;\param{LGTH}$.
\begin{definition}[Degeneracy]
An instance of Problem~\ref{prob:continuous} is called \emph{degenerate} if, for every sequence $\textbf{v}_1, \textbf{v}_2, \ldots \in \texttt{CLS}_\epsilon$ such that $\texttt{LGTH}(\textbf{v}_k) \to \texttt{LGTH}^*$ and $\textbf{v}^* \coloneqq \lim_{k \to \infty} \textbf{v}_k$ exists, the limit $\textbf{v}^*$ either (i) is not contained in $\texttt{CLS}_\epsilon$, (ii) belongs to a discontinuity set of $\texttt{LGTH}$, or (iii) has a first component (corresponding to the starting point of the closed trajectory) that is isolated in the set $\text{INL}_\epsilon$.
\label{def:degeneracy_closed}
\end{definition}
Finally, we introduce \emph{dense sampling procedures}.
\begin{definition}[Dense Sampling Procedure]
Suppose that, for each $N \in \mathbb{N}$, a set $\mathcal{A} \subset \real^2 \times S^1$ is sampled $N$ times to form a discrete subset $A_N$, i.e. $|A_N| = N$ for each $N$. Such a procedure is called a \emph{dense sampling procedure} if, for any $a \in \mathcal{A}$ and any open neighborhood $U$ of $a$, there exists $\hat{N} \in \mathbb{N}$ such that $A_N \cap U$ is non-empty for all $N > \hat{N}$.
\label{def:dense}
\end{definition}
%
%
We are now ready to rigorously characterize resolution completeness of Algorithm~\ref{alg:complete}.
\begin{theorem}[Resolution Completeness]
If Problem~\ref{prob:continuous} is not degenerate, and the set $\texttt{INL}_\epsilon$ is such that: (i) $\texttt{INL}_\epsilon \neq \emptyset$, and (ii) there exists $\hat{\jmath} \in \until M$ such that either $\texttt{INL}_\epsilon \subseteq \texttt{DWL}_{\hat{\jmath}}$ or $\texttt{DWL}_{\hat{\jmath}} \subseteq \texttt{INL}_\epsilon$, then Algorithm~\ref{alg:complete} is  \emph{resolution complete} in the following sense:

Form a sequence $\cramped{\{(\texttt{INL\_MNVR}_N, \texttt{CLS\_TRAJ}_N)\}_{N \in \mathbb{N}}}$, where each element $(\texttt{INL\_MNVR}_N, \texttt{CLS\_TRAJ}_N)$ represents the output of a call to Algorithm~\ref{alg:complete} when: (i) the number of discrete samples at each target is $N$, (ii) $\texttt{INL}_\epsilon^*$ (Algorithm~\ref{alg:discrete}, line $3$) is chosen to satisfy the conditions of Theorem~\ref{thm:gtsp}, and (iii) an optimal GTSP solution is found (Algorithm~\ref{alg:discrete}, line $5$). 
Then, $(\texttt{INL\_MNVR}_N, \texttt{CLS\_TRAJ}_N)$ is a feasible solution to Problem~\ref{prob:continuous} for each $N$, and if a dense sampling procedure is used to generate the discrete node sets at each target, 
then the length of the tours in the sequence $\{\texttt{CLS\_TRAJ}_N\}$ approaches the length of an optimal solution to Problem~\ref{prob:continuous} as $N \to \infty$.
\label{thm:resolution}
\end{theorem}
\begin{proof}
First note that the notion of resolution completeness described here is well-defined under the specified assumptions. Indeed, if the (non-discrete) set $\texttt{INL}_\epsilon$ satisfies the necessary assumptions, then the discrete analog always has the properties required for Theorem~\ref{thm:gtsp}, and $\texttt{INL}_\epsilon^*$ can be chosen to exploit the GTSP equivalence when constructing the sequence $\{(\texttt{INL\_MNVR}_N, \texttt{CLS\_TRAJ}_N)\}_{N \in \mathbb{N}}$.

Feasibility of each $\texttt{INL\_MNVR}_N$ follows from Theorem~\ref{thm:feasibility}, and the fact that feasible solutions to Problem~\ref{prob:graph} map to feasible solutions to Problem~\ref{prob:continuous}.
Let $\textbf{v}_1, \textbf{v}_2, \ldots \in \texttt{CLS}_\epsilon$ be a sequence such that $\texttt{LGTH}(\textbf{v}_i) \to \texttt{LGTH}^*$. Note that $\texttt{CLS}_\epsilon \subseteq (\real^2\times S^1)^M$ can be represented as a bounded subset in $(\real^{2}\times S^1)^M \subseteq \real^{4M}$. Invoking the Bolzano-Weierstrass theorem and compactness of $S^1$, any infinite sequence in $\texttt{CLS}_\epsilon$ contains a subsequence that converges to some point in $(\real^{2}\times S^1)^M$. Thus, recalling that Problem~\ref{prob:continuous} is non-degenerate, we can assume without loss of generality that $\textbf{v}^* = \lim_{i \to \infty} \textbf{v}_i$ exists, is contained in $\texttt{CLS}_\epsilon$, does not belong to a discontinuity sets of $\texttt{LGTH}$, and has a first component that is not isolated in $\param{INL}_\epsilon$. Since the function $\texttt{LGTH}$ is continuous at the point $\textbf{v}^*$, for any $\delta>0$, there exists an open neighborhood $U \subset \{(\param{DWL}_{\sigma(1)}\cap\param{INL}_\epsilon) \times \cdots\times \param{DWL}_{\sigma(M)}\;|\;\sigma \text{ permutes the set }\until M\}$ of $\textbf{v}^*$, within which $\texttt{LGTH}$ takes values within $\delta$ of $\texttt{LGTH}^*$. Since the sampling procedure is dense, for some $\hat{N}$, there will be a discrete node placed inside of the set $U$ for all $N\geq \hat{N}$. Since the GTSP is solved exactly, it follows that $|\texttt{LGTH}(\texttt{CLS\_TRAJ}_N) - \texttt{LGTH}^*| <\delta$ for $N >\hat{N}$, proving convergence. 
\end{proof}

Theorem~\ref{thm:resolution} states that, under a non-degeneracy
assumption about the problem structure, reasonable implementations of
Algorithm~\ref{alg:complete} will produce solutions that approximate
optimal solutions to the full, multi-objective routing
problem. Indeed, the resolution completeness property ensures that
Problem~\ref{prob:graph} will be appropriately reflective of
Problem~\ref{prob:continuous}.

A few final comments are in order. First, note that degenerate
problems occur only in very select situations: a degenerate problem instance is made non-degenerate by
perturbing target locations by an arbitrarily small amount (example
degenerate problem instances for particular cases are shown
in~\cite{KJO-PO-SD:12}). Thus, the non-degeneracy condition in
Theorem~\ref{thm:resolution} is not typically restrictive. Second, if
$\texttt{INL}_\epsilon$ does not satisfy the conditions of
Theorem~\ref{thm:resolution}, then resolution completeness does not
hold in general since Algorithm~\ref{alg:discrete} is not
guaranteed to find an optimal solution to Problem~\ref{prob:graph}
(Theorem~\ref{thm:gtsp}). Nevertheless, the quality of solutions
produced by Algorithm~\ref{alg:complete} generally increase as
sampling granularity is made increasingly fine. Finally, it is not
generally possible to find optimal solutions to GTSPs. Therefore, the
utility of Theorem~\ref{thm:resolution} is its ability to provide
intuition about qualitative solution behavior, and to ensure that
the discrete method herein is an appropriate approximation.


\section*{Acknowledgments}
The authors thank Karl J. Obermeyer of Obermeyer Labs for his
insight regarding vehicle routing.  

\bibliographystyle{asmems4}
\bibliography{alias,Main,New,FB}

\begin{thebibliography}{10}

\bibitem{JR:06}
Roberts, J., 2006.
\newblock ``Special issue on uninhabited aerial vehicles''.
\newblock {\em Journal of Field Robotics, {\bf 23}}(3--4).

\bibitem{KPV:08}
Valavanis, K.~P., 2008.
\newblock {\em Advances in Unmanned Aerial Vehicles: State Of The Art And the
  Road To Autonomy}.
\newblock Springer.

\bibitem{usaf:10}
{US Air Force}, 2010.
\newblock Report on technology horizons, a vision for {Air Force Science And
  Technology} during 2010--2030.
\newblock Tech. rep., AF/ST-TR-10-01-PR, United States Air Force.
\newblock Retrieved from
  http://www.defenseinnovationmarketplace.mil/\\resources/AF\_TechnologyHorizons2010-2030.pdf
  on Feb. 8, 2016.

\bibitem{osd:05}
{US~Office~of~the~Secretary~of~Defense}, 2005.
\newblock Unmanned aircraft systems ({UAS}) roadmap, 2005-2030.

\bibitem{USA:07}
Army, U., 2007.
\newblock Attack reconnaissance helicopter operations.
\newblock Tech. Rep. FM 3-04.126, Department of the Army (US), Feb.
\newblock Retrieved from
  http://usacac.army.mil/sites/default/files/misc/doctrine/\\CDG/fms.html on Dec.
  5, 2016.

\bibitem{KM:98}
Miettinen, K., 1998.
\newblock {\em Nonlinear Multiobjective Optimization}.
\newblock Springer.

\bibitem{KJO-PO-SD:12}
Obermeyer, K.~J., Oberlin, P., and Darbha, S., 2012.
\newblock ``Sampling-based path planning for a visual reconnaissance {UAV}''.
\newblock {\em {AIAA} Journal of Guidance, Control, and Dynamics, {\bf 35}}(2),
  pp.~619--631.

\bibitem{SR-RS-SD:05}
Rathinam, S., Sengupta, R., and Darbha, S., 2007.
\newblock ``A resource allocation algorithm for multi-vehicle systems with non
  holonomic constraints''.
\newblock {\em IEEE Transactions on Automation Sciences and Engineering, {\bf
  4}}(1), pp.~98--104.

\bibitem{DBK-RWB-RSH:08}
Kingston, D.~B., Beard, R.~W., and Holt, R.~S., 2008.
\newblock ``Decentralized perimeter surveillance using a team of {UAVs}''.
\newblock {\em IEEE Transactions on Robotics, {\bf 24}}(6), pp.~1394--1404.

\bibitem{JLF-PK-AG:15}
Las~Fargeas, J., Kabamba, P., and Girard, A., 2015.
\newblock ``Cooperative surveillance and pursuit using unmanned aerial vehicles
  and unattended ground sensors''.
\newblock {\em Sensors, {\bf 15}}(1), pp.~1365--1388.

\bibitem{VS-TS:15}
Shaferman, V., and Shima, T., 2015.
\newblock ``Tracking multiple ground targets in urban environments using
  cooperating unmanned aerial vehicles''.
\newblock {\em Journal of Dynamic Systems, Measurement, and Control, {\bf
  137}}(5), pp.~051010--051010--11.

\bibitem{CAR-CYS-AT:07}
Rabbath, C.~A., Su, C.~Y., and Tsourdos, A., 2007.
\newblock ``Guest editorial introduction to the special issue on multivehicle
  systems cooperative control with application''.
\newblock {\em IEEE Transactions on Control Systems Technology, {\bf 15}}(4),
  pp.~599--600.

\bibitem{JKH-BB-JL-BML:11}
Hedrick, J., Basso, B., Love, J., and Lavis, B., 2011.
\newblock ``Tools and techniques for mobile sensor network control''.
\newblock {\em Journal of Dynamic Systems, Measurement, and Control, {\bf
  133}}(2), pp.~024001--024001--7.

\bibitem{GG-APP:07}
Gutin, G., and Punnen, A.~P., 2007.
\newblock {\em The Traveling Salesman Problem and Its Variations}.
\newblock Springer.

\bibitem{KS-EF-FB:06h}
Savla, K., Frazzoli, E., and Bullo, F., 2008.
\newblock ``{T}raveling {S}alesperson {P}roblems for the {D}ubins vehicle''.
\newblock {\em IEEE Transactions on Automatic Control, {\bf 53}}(6),
  pp.~1378--1391.

\bibitem{JJE-KS-EF-FB:08k}
Enright, J.~J., Savla, K., Frazzoli, E., and Bullo, F., 2009.
\newblock ``Stochastic and dynamic routing problems for multiple {UAVs}''.
\newblock {\em {AIAA} Journal of Guidance, Control, and Dynamics, {\bf 34}}(4),
  pp.~1152--1166.

\bibitem{WM-SR-SD:07}
Malik, W., Rathinam, S., and Darbha, S., 2007.
\newblock ``An approximation algorithm for a symmetric generalized multiple
  depot, multiple travelling salesman problem''.
\newblock {\em Operations Research Letters, {\bf 35}}(6), pp.~747--753.

\bibitem{MN-PTK-ARG:15}
Niendorf, M., Kabamba, P.~T., and Girard, A.~R., 2016.
\newblock ``Stability of solutions to classes of traveling salesman problems''.
\newblock {\em IEEE Transactions on Cybernetics, {\bf 46}}(4).

\bibitem{MA-JPH:08}
Alighanbari, M., and How, J.~P., 2008.
\newblock ``A robust approach to the {UAV} task assignment problem''.
\newblock {\em International Journal on Robust and Nonlinear Control, {\bf
  18}}(2), pp.~118--134.

\bibitem{KH:00}
Helsgaun, K., 2000.
\newblock ``An effective implementation of the {L}in--{K}ernighan traveling
  salesman heuristic''.
\newblock {\em European Journal of Operational Research, {\bf 126}}(1),
  pp.~106--130.

\bibitem{VS-FP-FB:11za}
Srivastava, V., Pasqualetti, F., and Bullo, F., 2013.
\newblock ``Stochastic surveillance strategies for spatial quickest
  detection''.
\newblock {\em International Journal of Robotics Research, {\bf 32}}(12),
  pp.~1438--1458.

\bibitem{RP-PA-FB:14b}
Patel, R., Agharkar, P., and Bullo, F., 2015.
\newblock ``Robotic surveillance and {M}arkov chains with minimal weighted
  {K}emeny constant''.
\newblock {\em IEEE Transactions on Automatic Control, {\bf 60}}(12),
  pp.~3156--3167.

\bibitem{ym-ak:08}
Younis, M., and Akkaya, K., 2008.
\newblock ``Strategies and techniques for node placement in wireless sensor
  networks: {A} survey''.
\newblock {\em Ad Hoc Networks, {\bf 6}}(4), pp.~621--655.

\bibitem{FB-EF-MP-KS-SLS:10k}
Bullo, F., Frazzoli, E., Pavone, M., Savla, K., and Smith, S.~L., 2011.
\newblock ``Dynamic vehicle routing for robotic systems''.
\newblock {\em Proceedings of the IEEE, {\bf 99}}(9), pp.~1482--1504.

\bibitem{NM-SLS-SLW:15}
Mathew, N., Smith, S.~L., and Waslander, S.~L., 2015.
\newblock ``Multirobot rendezvous planning for recharging in persistent
  tasks''.
\newblock {\em IEEE Transactions on Robotics, {\bf 31}}(1), pp.~128--142.

\bibitem{FB-RC-PF:08u}
Bullo, F., Carli, R., and Frasca, P., 2012.
\newblock ``Gossip coverage control for robotic networks: {D}ynamical systems
  on the space of partitions''.
\newblock {\em SIAM Journal on Control and Optimization, {\bf 50}}(1),
  pp.~419--447.

\bibitem{JGC:12}
Carlsson, J.~G., 2012.
\newblock ``Dividing a territory among several vehicles''.
\newblock {\em INFORMS Journal on Computing, {\bf 24}}(4), pp.~565--577.

\bibitem{KJO:09}
Obermeyer, K., 2009.
\newblock ``Path planning for a uav performing reconnaissance of static ground
  targets in terrain''.
\newblock In AIAA Guidance, Navigation, and Control Conference, pp.~10--13.

\bibitem{JTI-JPH:13}
Isaacs, J., and Hespanha, J.~P., 2013.
\newblock ``{D}ubins traveling salesman problem with neighborhoods: {A}
  graph-based approach''.
\newblock {\em Algorithms, {\bf 6}}(1), pp.~84--99.

\bibitem{SML:06}
LaValle, S.~M., 2006.
\newblock {\em Planning Algorithms}.
\newblock Cambridge University Press.

\bibitem{PO-SR-SD:09b}
Oberlin, P., Rathinam, S., and Darbha, S., 2009.
\newblock ``A transformation for a heterogeneous, multiple depot, multiple
  traveling salesmen problem''.
\newblock In {A}merican {C}ontrol {C}onference, pp.~1292--1297.

\bibitem{CEN-JCB:91a}
Noon, C.~E., and Bean, J.~C., 1991.
\newblock ``A {L}agrangian based approach for the asymmetric generalized
  traveling salesman problem''.
\newblock {\em Operations Research, {\bf 39}}(4), pp.~623--632.

\bibitem{LVS-MSD:00}
Snyder, L.~V., and Daskin, M.~S., 2006.
\newblock ``A random-key genetic algorithm for the generalized traveling
  salesman problem''.
\newblock {\em European Journal of Operational Research, {\bf 174}}(1),
  pp.~38--53.

\bibitem{RTM-AJS:04}
Marler, R.~T., and Arora, J.~S., 2004.
\newblock ``Survey of multi-objective optimization methods for engineering''.
\newblock {\em Structural and multidisciplinary optimization, {\bf 26}}(6),
  pp.~369--395.

\bibitem{PJF-RCP-RJL:05}
Fleming, P.~J., Purshouse, R.~C., and Lygoe, R.~J., 2005.
\newblock ``Many-objective optimization: {A}n engineering design perspective''.
\newblock In {\em Evolutionary Multi-Criterion Optimization}, C.~A.~C. Coello,
  A.~H. Aguirre, and E.~Zitzler, eds. Springer, pp.~14--32.

\bibitem{DAG-DEJ:04}
Grundel, D., and Jeffcoat, D., 2004.
\newblock ``Formulation and solution of the target visitation problem''.
\newblock In Proceedings of the AIAA 1st Intelligent Systems Technical
  Conference.
\newblock AIAA 2004-6212.

\bibitem{LED:57}
Dubins, L.~E., 1957.
\newblock ``{O}n curves of minimal length with a constraint on average
  curvature and with prescribed initial and terminal positions and tangents''.
\newblock {\em American Journal of Mathematics, {\bf 79}}, pp.~497--516.

\end{thebibliography}
\end{document}